\documentclass[letterpaper, 10 pt, conference]{ieeeconf}  

\IEEEoverridecommandlockouts                              
\usepackage{amsmath} 
\usepackage{amssymb}  
\usepackage{amsthm}
\usepackage{bm}
\usepackage{color}
\usepackage{url}
\usepackage{graphicx}
\usepackage{epstopdf}
\usepackage[hidelinks]{hyperref}
\usepackage{mathtools}

\newcommand{\E}{\mathbb{E}}
\renewcommand{\Pr}{\mathbb{P}}

\newtheorem{thm}{Theorem}
\newtheorem{lemma}{Lemma}
\newtheorem{prop}{Proposition}

\theoremstyle{definition}
\newtheorem{assump}{Assumption}
\newtheorem{defn}{Definition}
\newtheorem{example}{Example}
\newtheorem{remark}{Remark}

\mathtoolsset{showonlyrefs}  
\allowdisplaybreaks  

\title{\LARGE \bf
Information Compression in Dynamic Information Disclosure Games
}

\author{Dengwang Tang and Vijay G. Subramanian
\thanks{A version of this work is submitted to IEEE CDC 2024.}
\thanks{This work is supported by NSF Grant No. ECCS 1750041, ECCS 2038416, ECCS 1608361, CCF 2008130, ARO Award No. W911NF-17-1-0232, MIDAS Sponsorship Funds by General Dynamics, and Okawa Foundation Research Grant.}
\thanks{D. Tang is with Ming Hsieh Department of Electrical and Computer Engineering,
        University of Southern California, Los Angeles, CA 90089, USA
        {\tt\small dengwang@usc.edu}}%
\thanks{V. G. Subramanian is with the Department of Electrical Engineering and Computer Science, University of Michigan,
        Ann Arbor, MI 48109, USA
        {\tt\small vgsubram@umich.edu}}%
}

\begin{document}

\maketitle

\thispagestyle{plain}  
\pagestyle{plain}  

\begin{abstract}
 We consider a two-player dynamic information design problem between a principal and a receiver---a game is played between the two agents on top of a Markovian system controlled by the receiver's actions, where the principal obtains and strategically shares some information about the underlying system with the receiver in order to influence their actions. In our setting, both players have long-term objectives, and the principal sequentially commits to their strategies instead of committing at the beginning. Further, the principal cannot directly observe the system state, but at every turn they can choose randomized experiments to observe the system partially. The principal can share details about the experiments to the receiver. For our analysis we impose the \emph{truthful disclosure} rule: the principal is required to truthfully announce the details and the result of each experiment to the receiver immediately after the experiment result is revealed. Based on the received information, the receiver takes an action when its their turn, with the action influencing the state of the underlying system. We show that there exist Perfect Bayesian equilibria in this game where both agents play Canonical Belief Based (CBB) strategies using a compressed version of their information, rather than full information, to choose experiments (for the principal) or actions (for the receiver). We also provide a backward inductive procedure to solve for an equilibrium in CBB strategies.
\end{abstract}

\section{INTRODUCTION}
In many modern engineering and socioeconomic problems and systems, such as cyber-security, transportation networks, and e-commerce, information asymmetry is an inevitable aspect that crucially impacts decision making. In these systems, agents need to decide on their actions under limited information about the system and each other. In many situations, agents can overcome (some of) the information asymmetry by communicating with each other. However, agents can be unwilling to share information when agents' goals are not aligned with each other, since having some information that another agent does not know can be an advantage. In general, communication between agents with diverging incentives cannot be naturally established without rules/protocols that everyone agrees upon, and all agents suffer due to the breakdown of the information exchange. For example, drug companies are required by regulations to disclose their trial results truthfully. The public can then trust the results and benefit from the drug. In turn, the drug companies can make a profit. Without government regulations, drug companies and the public will both suffer due to mistrust. In many real-world dynamic systems, information exchange and decision making can happen repeatedly as the system/environment changes over time---for example, public companies disclose information periodically which impacts stockholders' decisions; (COVID-19) vaccine producers conduct their trials and release results sequentially which impacts the government's purchasing decisions; during an epidemic, health authorities update their recommendations on the use of face masks over time according to changing levels of infection, etc. Therefore, in the face of information asymmetry, it is important to establish rules/protocols to facilitate repeated information exchange among agents in multi-agent dynamic systems. 

In the economics literature, there are mainly two approaches to the above problem, namely mechanism design \cite{myerson1989mechanism} and information design \cite{kamenica2011bayesian}. In mechanism design, less informed agents can extract information from more informed agents by committing to how they will use the collected information beforehand. Whereas in information design, more informed agents can partially disclose information to less informed agents. The more informed agents commit on the manner in which they partially disclose their information. In both approaches, all agents can benefit from the information exchange.
For both approaches, one can classify the pertinent literature into two groups: (i) static settings, where both information disclosure and decision making take place only once; and (ii) dynamic settings, where agents repeatedly disclose information and take actions over time on top of an ever changing environment/physical system. Mechanism design and information design for the dynamic settings are more challenging than the static settings since agents need to anticipate future information disclosure when taking an action. Mechanism design in dynamic settings has been studied extensively in the literature \cite{bergemann2010dynamic,athey2013efficient,pavan2014dynamic,bergemann2019dynamic}.
In most of the works on information design in dynamic settings, the receivers are assumed to be myopic \cite{lingenbrink2017optimal,ely2017beeps,farokhi2016estimation,sayin2016strategic,renault2017optimal,best2016honestly,best2016persuasion}. This assumption greatly simplifies receivers' decision making. There have been a few papers studying information design problems where all agents in the system have long-term goals \cite{LUH1984251,4047637,tolwinski1981closed,farhadi2018static,8619619,sayin2018dynamic,sayin2019hierarchical,sayin2019optimality,meigs2020optimal,tavafoghi2017informational,farhadi2022dynamic}. These papers typically assume that the principal commits to their strategy for the whole game before the game starts as in a Stackelberg equilbrium~\cite{heinrich2011market}. The bulk of this literature also assumes that the principal observes the underlying state perfectly. However, these assumptions can be inappropriate for many applications. If the protocol gives more informed agents the power to commit to a strategy for the whole time horizon at the beginning of time, then the more informed agent can implement punishment strategies by threatening to withhold information if less informed agents do not obey their ``instructions''---see Example \ref{ex:did:threatbased}. Thus, the informed agents could abuse their commitment power to implement otherwise non-credible threats instead of using it for efficient information disclosure. This is not a desirable outcome---for example, online map services should not threaten to withhold service if a driver refuses to take the recommended route; and similarly, public health authorities may want to use persuasion instead of threats to encourage mask wearing during an epidemic. Again, focusing on the public health setting, during an epidemic the authorities may not know the full extent of the disease spread, but only an estimate of it (using testing and other methods). In this context, transparency to the public on the part of the authorities---in disclosing measurement methods and data---is important for persuasion based schemes to be effective.

In this work, we focus on the dynamic information design problem. Specifically, we consider a dynamic game between a principal and a receiver on top of a Markovian system. Both the principal and the receiver have long-term objectives. The principal cannot directly (and perfectly) observe~\footnote{Generalizing to direct and perfect observations is technically challenging as existence of Nash equilibria and sequential refinements thereof, with or without compressing information, for infinite state or action dynamic games~\cite{myerson2020perfect} is non-trivial or may not hold in great generality.} the system state, but can choose randomized experiments to partially observe the system. The principal is also allowed to choose any experiment, but they must announce the experimental setup and results truthfully to the receiver before the receiver takes their action. Both these aspects of our model are motivated by the public health setting described earlier. The receiver takes action on each turn based on the information received to date, which then influences the underlying system. 

\textbf{Contributions:} In the class of dynamic information disclosure games among a principal and a receiver discussed above, under the assumption of truthful disclosure, we identify compression based strategies, called Canonical Belief Based (CBB) strategies, for both players to play at equilibrium. Here both agents use strategies based on a compressed version of their information, rather than the full information, to choose their actions---experiments (for the principal), and actions (for the receiver). We develop a backward inductive sequential decomposition procedure to find such equilibrium strategies, and we show that the procedure always has at least one solution. Finally, we investigate examples of such games to provide insight to CBB-strategy-based equilibria. 

\textbf{Organization:}
The rest of this work is organized as follows: In Section \ref{sec:did:motex}, we provide an example where a principal can abuse its commitment power. We formulate the problem in Section \ref{sec:dyninfo:setting}. Then, we provide some preliminary results in discrete geometry in Section \ref{sec:did:geom}. In Section \ref{sec:did:mainresults} we state our main results. In Section \ref{sec:did:examples} we study some examples. We discuss difficulties with potential extensions of our result in Section \ref{sec:did:dis}. Finally, we conclude in Section \ref{sec:concl}. Details of all the technical proofs are in the Appendix.

\textbf{Notation:} We use capital letters to represent random variables, bold capital letters to denote random vectors, and lower case letters to represent realizations. We use superscripts to indicate teams and agents, and subscripts to indicate time. We use $t_1:t_2$ to indicate the collection of timestamps $(t_1, t_1+1, \cdots, t_2)$. For example $X_{1:3}^1$ stands for the random vector $(X_1^1, X_2^1, X_3^1)$. For random variables or random vectors, we use the corresponding script capital letters (italic capital letters for greek letters) to denote the space of values these random vectors can take. For example, $\mathcal{H}_t^i$ denotes the space of values the random vector $H_t^i$ can take. We use $\Pr(\cdot)$ and $\E[\cdot]$ to denote probabilities and expectations, respectively. We use $\Delta(\varOmega)$ to denote the set of probability distributions on a finite set $\varOmega$.

\subsection{A Motivating Example}\label{sec:did:motex}
The following is an example where given the power to commit to a strategy for the whole game before the game starts (i.e. the Stackelberg game setting), the principal can use otherwise non-credible threats.

\begin{example}\label{ex:did:threatbased}
	Consider a two-stage game of two players: the principal $P$, and the agent/receiver $R$. The state of the system at time $t$ is $X_t$. The states are uncontrolled, and $X_1, X_2$ are i.i.d. uniform random variables taking values in $\{0, 1\}$. The principal can observe $X_t$ at time $t$ while the receiver cannot. At stage $t$, The principal transmits message $M_t$ to the receiver and the receiver takes an action $U_t\in  \{a, b, c, d\}$. The instantaneous payoff for both players are given by
	\begin{align*}
		r_1^A(0, a) = 1, r_1^A(0, b) = 1.01, r_1^A(0, c) = r_1^A(0, d) = -1000\\
		r_1^A(1, c) = 1, r_1^A(1, d) = 1.01, r_1^A(1, a) = r_1^A(1, b) = -1000\\
		r_1^B(0, a) = 500, r_1^B(0, b) = 1, r_1^B(0, c) = r_1^B(0, d) = -1000\\
		r_1^B(1, c) = 500, r_1^B(1, d) = 1, r_1^B(1, a) = r_1^B(1, b) = -1000
	\end{align*}
	and $r_2^A(\cdot, \cdot)=r_2^B(\cdot, \cdot)=r_1^A(\cdot, \cdot)$.
	
	Suppose that the principal has the power to commit to a strategy $(g_1, g_2)$ at the beginning of the game. Then, (given the Stackelberg setting) an optimal  strategy for the principal is the following: fully reveal the state at $t=1$ (i.e. $M_1=X_1$); if the receiver plays $a$ or $c$ at $t=1$, then transmit no information at $t=2$; and if the receiver plays $b$ or $d$ at $t=1$, then fully reveal the state at $t=2$. Then, the receiver's best response to the principal's strategy is the following: at $t=1$: play $b$ if $M_1=0$, and play $d$ if $M_1=1$; and at time $t=2$: play $a$ if $M_2=0$, and play $c$ if $M_2=1$. 
 
    In the resulting equilibrium, the principal effectively \emph{threatens} the receiver to comply to their interest at time $t=1$ by not giving information at time $t=2$, even though the interests of both parties are aligned at $t=2$. In fact, without posing a threat to the receiver at time $2$, the principal cannot convince them to play $b$ or $d$ at time $1$.
\end{example}

\section{PROBLEM FORMULATION}\label{sec:dyninfo:setting}
We consider a finite-horizon two players dynamic game between the principal $A$ and the agent/receiver $B$. The game consists of $T$ stages, where, in each stage, the principal moves before the receiver. The game features an underlying dynamic system with state $X_t$.
At each time $t\in [T]$, the receiver chooses an action $U_t$. Then, the system transits to the next state $X_{t+1}\sim P_t(X_t, U_t)$, where $P_t:\mathcal{X}_t\times\mathcal{U}_t\mapsto \Delta(\mathcal{X}_{t+1})$ is the transition kernel. The initial state $X_1$ has prior distribution $\hat{\pi}\in\Delta(\mathcal{X}_1)$. The initial distribution $\hat{\pi}$ and transition kernels $P = (P_t)_{t=1}^{T}$ are common knowledge to both players.
We assume that neither player can observe the state $X_t$ directly. However, at each time $t$, the principal can conduct an \emph{experiment} to learn about $X_t$. In this work, an experiment\footnote{Information design literature~\cite{gentzkow2017bayesian,gentzkow2017disclosure} deems such experiments \emph{signals}.} refers to an observation kernel that can be chosen by the principal. We impose the rule that the experiments are required to be public---both the principal and the receiver know the settings (the probabilities in the observation kernel), and the outcome (the observation itself) of the experiment. Specifically, at each time $t$, the principal chooses an observation kernel $\sigma_t:\mathcal{X}_t\mapsto \Delta(\mathcal{M}_t)$, and announces $\sigma_t$ to the receiver. The experiment outcome $M_t$ is then realized, and observed by both the principal and the receiver. Note that if the horizon $T=1$, then the above setting is the same as the classical information design problem considered in \cite{kamenica2011bayesian}.

\begin{assump}
	$\mathcal{X}_t, \mathcal{U}_t, \mathcal{M}_t$ are finite sets with $|\mathcal{M}_t|$ sufficiently large.
\end{assump}

The order of events happening at time $t$ is given as the following:
(1) The principal commits to an experiment $\sigma_t$, and announces it to the receiver;
(2) The measurement result $M_t$ is revealed to both the principal and receiver; 
(3) The receiver takes action $U_t$; and 
(4) $X_t$ transits to the next state.

Let $\mathcal{S}_t$ be the space of experiments. The principal uses a (pure) strategy to choose their experiment $g_t^A: \mathcal{S}_{1:t-1} \times \mathcal{M}_{1:t-1} \times \mathcal{U}_{1:t-1} \mapsto \mathcal{S}_t$. For convenience, define $\mathcal{H}_t^A=\mathcal{S}_{1:t-1} \times \mathcal{M}_{1:t-1} \times \mathcal{U}_{1:t-1}$.
The receiver uses a (pure) strategy $g_t^B: \mathcal{S}_{1:t} \times \mathcal{M}_{1:t} \times \mathcal{U}_{1:t-1} \mapsto \mathcal{U}_t$. For convenience, define $\mathcal{H}_t^B=\mathcal{S}_{1:t} \times \mathcal{M}_{1:t} \times \mathcal{U}_{1:t-1}$.
The principal's goal is to maximize $J^A(g) = \E^g\left[\sum_{t=1}^{T} r_t^A(X_t, U_t) \right]$. The receiver's goal is to maximize $J^B(g)=\E^g\left[\sum_{t=1}^{T} r_t^B(X_t, U_t) \right]$. The instantaneous reward functions $(r_t^A, r_t^B)_{t=1}^T$ are common knowledge to both agents.

The belief of the principal at time $t$ is a function $\mu_t^A: \mathcal{M}_{1:t-1} \times \mathcal{S}_{1:t-1} \times \mathcal{U}_{1:t-1} \mapsto \Delta(\mathcal{X}_{1:t})$. The belief of the receiver at time $t$ (after knowing $\sigma_t$ and observing $M_t$) is a function $\mu_t^B: \mathcal{M}_{1:t} \times \mathcal{S}_{1:t} \times \mathcal{U}_{1:t-1} \mapsto \Delta(\mathcal{X}_{1:t})$.

Inspired by the ``mechanism picking game'' defined in \cite{doval2018mechanism}, we call the above game a \emph{signal picking game}, and we will study Perfect Bayesian Equilibria for our game. 
\begin{defn}[PBE]
	A Perfect Bayesian Equilibrium is a pair $(g, \mu)$, where
	\begin{itemize}
		\item $g$ is sequentially rational given $\mu$ ($=\big(\mu^A_{1:T},\mu^B_{1:T}\big)$).
		\item $\mu$ can be updated using Bayes law whenever the denominator is non-zero.
	\end{itemize}
\end{defn}

\section{BACKGROUND: DISCRETE GEOMETRY}\label{sec:did:geom}
In this section, we introduce some notations and results of discrete geometry that are necessary for our main results. 

\begin{defn}\label{def:did:piecewiselinear}
	Let $f$ be a real-valued function on a polytope\footnote{A polytope is a convex hull of a finite set in $\mathbb{R}^d$ where $d < +\infty$.} $\Omega$. Then, $f$ is called a (continuous) piecewise linear function if there exist polytopes $C_1, \cdots, C_k$ such that
	\begin{itemize}
		\item $f$ is linear on each $C_j$ for $j=1,\cdots,k$; and
		\item $C_1\cup\cdots \cup C_k = \Omega$.
	\end{itemize}
\end{defn}

\begin{lemma}\label{lem:did:plcompose}
	Let $\Omega_1, \Omega_2$ be polytopes. Let $\ell:\Omega_1\mapsto \Omega_2$ be an affine function and $f:\Omega_2\mapsto \mathbb{R}$ be a piecewise linear function. Then the composite function $f\circ \ell: \Omega_1\mapsto \mathbb{R}$ is piecewise linear.
\end{lemma}
\begin{proof}
    See Appendix \ref{app:disgeo}.
\end{proof}


Next, we introduce the notion of a triangulation. 

\begin{defn}
	\cite{de2010triangulations} Let $\Omega$ be a finite dimensional polytope. A \emph{triangulation} $\gamma$ of $\Omega$ is a finite collection of simplices (i.e. convex hulls of a finite, affinely independent set of points) such that 
 \begin{enumerate}
     \item[(1)] If a simplex $C \in\gamma$, then all faces of $C$ are in $\gamma$; 
     \item[(2)] For any two simplices $C_1, C_2\in \gamma$, $C_1\cap C_2$ is a (possibly empty) face of $C_1$; and 
     \item[(3)] The union of all simplices in $\gamma$ equals $\Omega$.
 \end{enumerate}
\end{defn}

\begin{figure}[!ht]
	\centering
	\includegraphics[width=6cm]{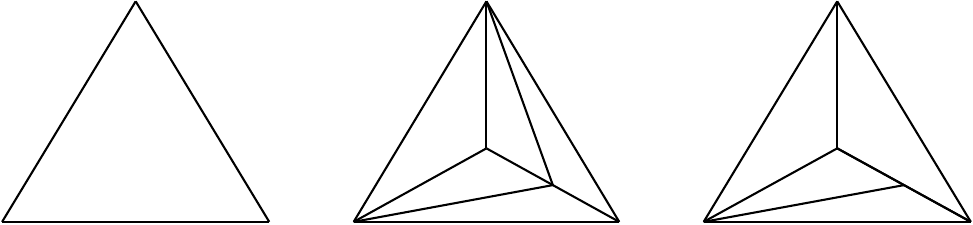}
	\caption{Left: 2-D Polytope $\Omega$; Center: A triangulation; Right: NOT a triangulation.}
\end{figure}

For a function $f:\Omega\mapsto\mathbb{R}$ and a triangulation $\gamma$, let $\mathbb{I}(f, \gamma)$ denote the linear interpolation of $f$ based on the triangulation $\gamma$, i.e.
\begin{equation}
	\mathbb{I}(f, \gamma)(\omega) := \alpha_1 f(\omega_1) + \cdots + \alpha_k f(\omega_k),
\end{equation}
if $\omega\in C$, where $C\in \gamma$ is a simplex with vertices $\omega_1, \cdots, \omega_k$, and $\omega = \alpha_1 \omega_1 + \cdots + \alpha_k \omega_k$ for some $\alpha_1, \cdots, \alpha_k \geq 0$ such that $\alpha_1+ \cdots +\alpha_k = 1$.

\begin{figure}[!ht]
	\centering
	\includegraphics[width=6cm]{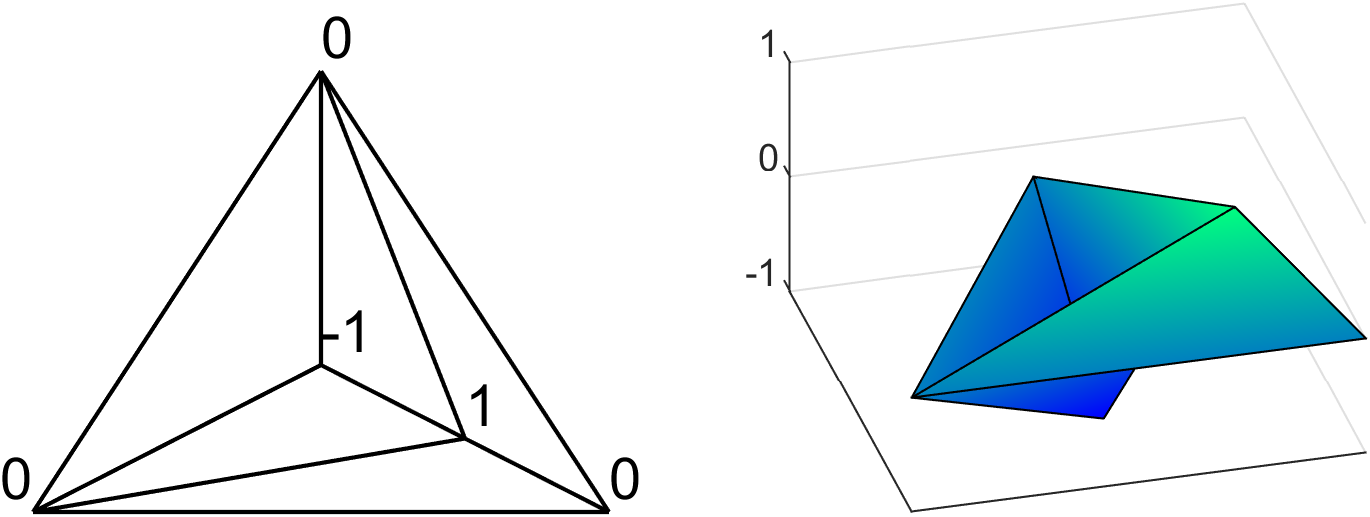}
	\caption{Left: A triangulation $\gamma$ labeled with the values of a function $f$ on the vertices. Right: 3-D plot of $\mathbb{I}(f, \gamma)$.}
\end{figure}

\begin{lemma}\label{lem: intpcont}
	For any real-valued function $f$ on a polytope $\Omega$, $\mathbb{I}(f, \gamma)$ is a well-defined, continuous piecewise linear function.
\end{lemma}
\begin{proof}
    See Appendix \ref{app:disgeo}.
\end{proof}

For each $\omega\in\Omega$ and triangulation $\gamma$, we have shown that there exists a unique way to represent $\omega$ as a convex combination of the vertices of one simplex from $\gamma$. One can treat this convex combination as a finite measure. Denote this finite measure by $\mathbb{C}(\omega, \gamma)$. Then we have $\mathbb{I}(f, \gamma)(\omega) = \int f(\cdot) \mathrm{d}\mathbb{C}(\omega, \gamma)$.


\begin{defn}
	Let $f$ be a real-valued function on $\Omega$. The concave closure $\mathrm{cav}(f)$ of $f$ is defined as a function $\rho$ such that  
	\begin{equation}
		\rho(\omega) := \sup \{z: (\omega, z)\in \mathrm{cvxg}(f) \} \quad\forall \omega\in\Omega
	\end{equation}
	where $\mathrm{cvxg}(f)\subset \Omega \times \mathbb{R}$ is the convex hull of the graph of $f$.
\end{defn}

For certain functions $f$, their concave closures can be represented as a triangulation based interpolation of the original function. Define the set of all such triangulations as $\arg\mathrm{cav}(f)$, i.e.
\begin{equation}
\begin{split}
 \arg\mathrm{cav}(f) 
 := 
 \{\gamma \text{ is a triangulation of }\Omega : \mathbb{I}(f, \gamma) = \mathrm{cav}(f) \}.
\end{split}
\end{equation} 

\begin{figure}[!ht]
	\centering
	\includegraphics[height=2cm]{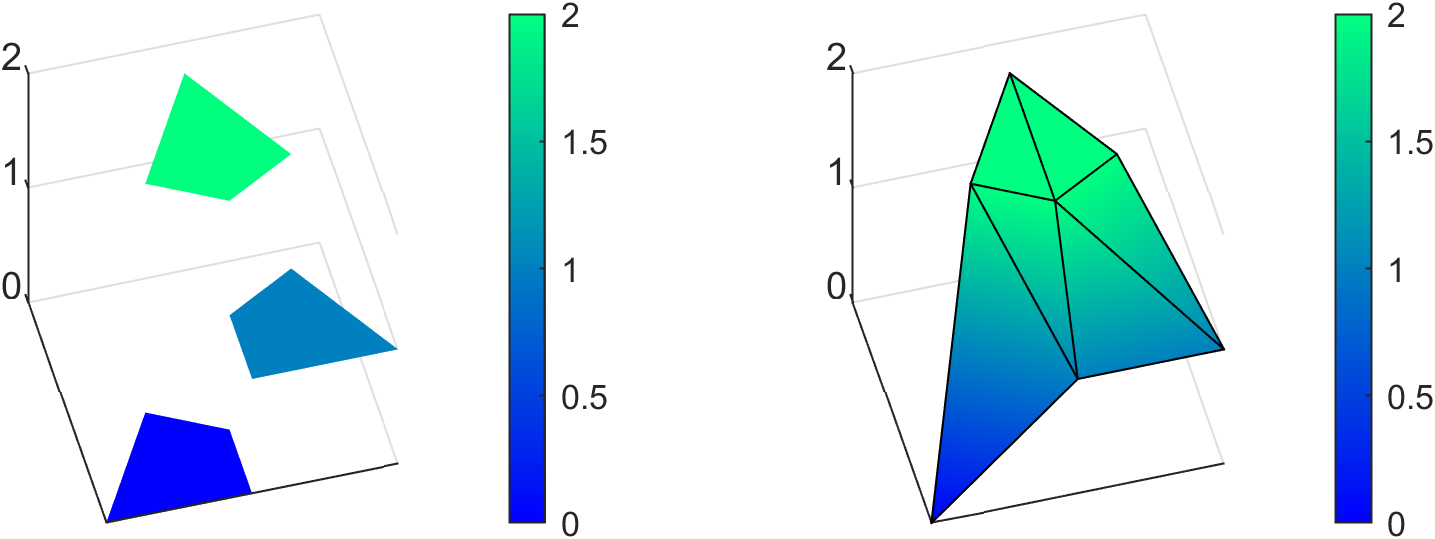}\\\vspace{1em}
	\includegraphics[height=1.9cm]{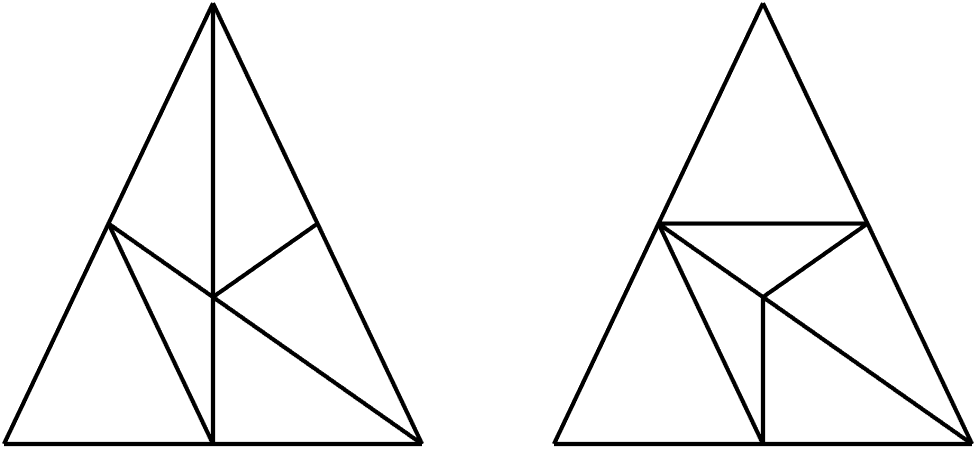}
	\caption{Top-left: 3-D plot of a function $f$ (an upper semi-continuous piecewise constant function taking values in $\{0, 1, 2\}$). Top-right: Concave closure of $f$. Bottom-left and bottom-right: 2-D visualization of two different triangulations in $\arg\mathrm{cav}(f)$.}
\end{figure}

The following lemma identifies a class of functions with the above property.

\begin{lemma}\label{lem:did:cavmax}
	Let $f_1, \cdots, f_k, \rho_1, \cdots, \rho_k$ be continuous piecewise linear functions on a polytope $\Omega$. For $\omega\in\Omega$, define
	\begin{align*}
		\Upsilon(\omega) = \underset{j=1,\cdots, k}{\arg\max} f_j(\omega), \text{ and }
		\Psi(\omega) = \max_{j\in \Upsilon(\omega)} \rho_j(\omega).
	\end{align*}
	Then $\arg\mathrm{cav}(\Psi)$ is non-empty, i.e. there exists a triangulation $\gamma$ of $\Omega$ such that the concave closure of $\Psi$ is equal to $\mathbb{I}(\Psi, \gamma)$.
\end{lemma}
\begin{proof}
    See Appendix \ref{app:disgeo}.
\end{proof}

\section{MAIN RESULTS}\label{sec:did:mainresults}
In this section, we introduce our main result---Theorem \ref{thm:dyninfodesign}---, which provides a dynamic programming characterization of a subset of PBE of the signal picking game.

Note that due to the assumption of public experiments, the signal picking game is a game with symmetric information~\footnote{As mentioned earlier---in Footnote 1---, there are significant challenges is generalizing to asymmetric information settings.} after each experiment is conducted. The principal's advantages lies in the fact that they have the power to determine the choice of experiments. Thus, standard results on strategy-independence of beliefs (e.g. \cite{kumar1986stochastic}) imply that the beliefs of both players in this game are strategy-independent, i.e. there is a canonical belief system. Similar strategy-independent belief systems are also constructed and used in \cite{nayyar2013common}. We describe this belief system as follows. 

\begin{defn}\label{def:dyninfo:bayesupdate}
	Define the Bayesian update function $\xi_t: \Delta(\mathcal{X}_t)\times \mathcal{S}_t\times \mathcal{M}_t\mapsto \Delta(\mathcal{X}_t)$ by setting for each $x_t\in \mathcal{X}_t$
	\begin{equation}
		\xi_t(x_{t}|\pi_t, \sigma_t, m_t) := \dfrac{\pi_{t}(x_t) \sigma_t(m_t|x_t) }{\sum_{\tilde{x}_t} \pi_{t}(\tilde{x}_t) \sigma_t(m_t|\tilde{x}_t) }
	\end{equation}
	for all $(\pi_t, \sigma_t, m_t)$ such that the denominator is non-zero. When the denominator is zero, $\xi_t(\pi_t, \sigma_t, m_t)$ is defined to be the uniform distribution.
\end{defn}

\begin{defn}\label{def:canonicalbelief}
	The canonical belief system is a collection of functions $(\kappa_t^A, \kappa_t^B)_{t\in\mathcal{T}}, \kappa_t^i: \mathcal{H}_t^i \mapsto \Delta(\mathcal{X}_t), i\in\{A, B\}$ defined recursively through the following step. Denote $\pi^i_t=\kappa_t^i(h_t^i), i\in\{A, B\}, t\in\mathcal{T}$, and then we have
	\begin{itemize}
		\item $\pi_1^A:=\hat{\pi}$, the prior distribution of $X_1$;
		\item $\pi_{t}^B:= \xi_t(\pi_t^A, \sigma_t, m_t)$;
		\item $\pi_{t+1}^A := \ell_t(\pi_{t}^B, u_{t})$, where $\ell_t: \Delta(\mathcal{X}_t)\times \mathcal{U}_t \mapsto \Delta(\mathcal{X}_{t+1})$ is defined by
		\begin{equation}
			\ell_t(\pi_t, u_{t})(x_{t+1}) := \sum_{\tilde{x}_t}  \pi_t(\tilde{x}_t) P_t(x_{t+1}|\tilde{x}_t, u_t).
		\end{equation}
	\end{itemize}
\end{defn}

We consider a subclass of strategies for both the principal and the receiver, called canonical belief based (CBB) strategies, wherein player $i\in \{A,B\}$ chooses their experiment or action, respectively, at time $t$ based solely on $\Pi_t^i=\kappa_t^i(H_t^i)$ instead of $H_t^i$. Let $\lambda_t^A: \Delta(\mathcal{X}_t) \mapsto \mathcal{S}_t$ be the CBB strategy of the principal, and $\lambda_t^B:\Delta(\mathcal{X}_t) \mapsto \mathcal{U}_t$ be the CBB strategy of the receiver. Then, saying that player $i$ is using CBB strategy $\lambda_t^i$ is equivalent to saying that they are using the strategy
\begin{equation*}
	g_t^i(h_t^i) = \lambda_t^i(\kappa_t^i(h_t^i)),\qquad\forall h_t^i\in\mathcal{H}_t^i.
\end{equation*}

Given an experiment and a distribution on the state, the posterior belief of the receiver is a random variable (a function of the random outcome). In an information disclosure game, it is helpful to consider the following sub-problem: how to design an experiment such that the receiver's  belief, as a random variable, follows a certain distribution. The next definition formalizes this concept. This concept was used in classical one-shot information design setting \cite{kamenica2011bayesian} as well.

\begin{defn}
	\cite{kamenica2011bayesian} An experiment $\sigma_t\in \mathcal{S}_t$ is said to \emph{induce a distribution} $\eta\in \Delta_f(\Delta(\mathcal{X}_t))$---that is, $\eta$ is a distribution with finite support on the set of distributions $\Delta(\mathcal{X}_t)$---\emph{from} $\pi_t\in \Delta(\mathcal{X}_t)$ \cite{kamenica2011bayesian} if for all $\tilde{\pi}_{t}\in \Delta(\mathcal{X}_t)$, 
	\begin{equation}
		\eta(\tilde{\pi}_{t}) = \sum_{\tilde{m}_{t}}  \bm{1}_{\{\tilde{\pi}_{t} = \xi_{t}(\pi_{t}, \sigma_{t}, \tilde{m}_{t}) \} } \sum_{\tilde{x}_{t}}\sigma_{t}(\tilde{m}_{t}|\tilde{x}_{t}) \pi_{t}(\tilde{x}_{t}).
	\end{equation}
	
	A distribution $\eta$ is said to be \emph{inducible from} $\pi_t$ if there exists some experiment $\sigma_{t}$ that induces $\eta$ from $\pi_t$. 
\end{defn}

\begin{remark}
	In \cite{kamenica2011bayesian}, the authors showed that a distribution is $\eta\in \Delta_f(\Delta(\mathcal{X}_t))$ is inducible from $\pi_t$ if and only if $\pi_t$ is the center of mass of $\eta$, i.e. $\pi_t = \sum_{\tilde{\pi}_t\in \mathrm{supp}(\eta)} \eta(\tilde{\pi}_t)\cdot \tilde{\pi}_t$. 
\end{remark}

We now introduce our main result, which describes a backward induction procedure to find a PBE where both players use CBB strategies.

\begin{thm}\label{thm:dyninfodesign}
	Let
	\begin{align*}
		V_{T+1}^A(\cdot) = V_{T+1}^B(\cdot) := 0
	\end{align*}
	
	For each $t=T, T-1, \cdots, 1$ and $\pi_t\in \Delta(\mathcal{X}_t)$, define
	\begin{subequations}\label{eq:dyninfo:dynprog}
		\begin{align}
			\hat{q}_t^i(\pi_t, u_t) &:= \sum_{\tilde{x}_t} r_t^i(\tilde{x}_t, u_t)\pi_t(\tilde{x}_t)  + V_{t+1}^i(\ell_t(\pi_t, u_t))\\&\qquad\qquad\quad \forall i\in \{A, B\};\label{eq:dyninfo:dynprog:1}\\
			\Upsilon_t(\pi_t) &:= \underset{u_t}{\arg\max}~ \hat{q}_t^B(\pi_t, u_t);\label{eq:dyninfo:dynprog:2} \\
			\hat{v}_t^A(\pi_t) &:= \underset{u_t\in \Upsilon(\pi_t)}{\max}~ \hat{q}_t^A(\pi_t, u_t);\label{eq:dyninfo:dynprog:vta}\\
			\hat{v}_t^B(\pi_t) &:= \max_{u_t} \hat{q}_t^B(\pi_t, u_t);\label{eq:dyninfo:dynprog:vtb}\\ 
			\gamma_t &\in \arg\mathrm{cav}(\hat{v}_t^A);\label{eq:dyninfo:dynprog:4}\\
			V_t^i(\pi_t) &:= \mathbb{I}(\hat{v}_t^i, \gamma_t)\quad \forall i\in \{A, B\}.\label{eq:dyninfo:dynprog:5}
		\end{align}
	\end{subequations}
	Let $\lambda_t^{*B}(\pi_t)$ be any $u_t\in\mathcal{U}_t$ that attains the maximum in \eqref{eq:dyninfo:dynprog:vta}.
	Let $\lambda_t^{*A}(\pi_t)$ be any experiment that induces the finite measure $\mathbb{C}(\pi_t, \gamma_t)$ from $\pi_t$. Then, the CBB strategies $(\lambda^{*A}, \lambda^{*B})$ form (the strategy part of) a PBE, and $V_1^A(\hat{\pi})$ and $V_1^B(\hat{\pi})$ are the equilibrium payoffs for the principal and the receiver respectively in this PBE.
\end{thm}
\noeqref{eq:dyninfo:dynprog:1,eq:dyninfo:dynprog:2,eq:dyninfo:dynprog:vta,eq:dyninfo:dynprog:vtb,eq:dyninfo:dynprog:4,eq:dyninfo:dynprog:5}

\begin{proof}[Proof Outline]
    In Lemma \ref{lem:canonicalbelief} in Appendix \ref{app:main} we construct a belief system $\mu^*$ that is consistent with any strategy profile. Hence, we only need to show sequential rationality of $\lambda^*$. 
    
    To show the receiver's sequential rationality, we prove the following: Fixing the principal's strategy to be $\lambda^{*A}$, the receiver is facing an MDP with state $\Pi_t^B$ and action $U_t$. The proof then follows via standard stochastic control arguments.

    To show the principal's sequential rationality, we prove the following: Fixing the receiver's strategy to be $\lambda^{*B}$, the principal is facing an MDP with state $\Pi_t^A$ and action $\Sigma_t$. This proof follows cia standard stochastic control arguments coupled with information design results~\cite{kamenica2011bayesian}.

    The details of the proof are presented in Appendix \ref{app:main}.
\end{proof}

The following proposition states that the sequential decomposition procedure described in Theorem \ref{thm:dyninfodesign} is well defined and always has a solution.
\begin{prop}\label{lem: welldefined}
	There always exists a CBB strategy profile $(\lambda^{*A}, \lambda^{*B})$ that satisfies Eqs. \eqref{eq:dyninfo:dynprog} in Theorem \ref{thm:dyninfodesign}.
\end{prop}

\begin{proof}
    Induction on time $t$ is used for the proof.
 
	\textbf{Induction Invariant}: $V_{t}^A, V_{t}^B$ are well-defined continuous piecewise linear functions.
	
	\textbf{Induction Base}: The induction variant is clearly true for $t=T+1$ since $V_{T+1}^A, V_{T+1}^B$ are constant functions.
	
	\textbf{Induction Step}: Suppose that the induction invariant holds for $t+1$.
	\begin{itemize}
		\item Step 1: For each $u_t\in\mathcal{U}_t$, using the fact that $\ell_t(\pi_t, u_t)$ is affine in $\pi_t$, by Lemma \ref{lem:did:plcompose}, $q_t^A, q_t^B$ are continuous piecewise linear functions in $\pi_t$.
		
		\item Step 2: By Lemma \ref{lem:did:cavmax}, $\gamma_t$ is well-defined.
		
		\item Step 3: By Lemma \ref{lem: intpcont}, $V_t^A, V_t^B$ are continuous piecewise linear functions.
	\end{itemize}
 This completes the proof.
\end{proof}

\subsection{Extension}
In many real-world settings, the receivers have the option to quit the game at any time. Our model and results can be extended to finite horizon games where the receiver can decide to terminate the game at any time before time $T$.  

\begin{prop}\label{prop:dyninfodesignwquit}
	Let $\overline{\mathcal{U}}_t\subset \mathcal{U}_t$ be the set of actions that terminates the game at time $t$. If we define $V_t^i, q_t^i, \lambda_t^{*i}$ for each $i\in\{A, B\}, t\in\mathcal{T}$ as in \eqref{eq:dyninfo:dynprog}
	except that \eqref{eq:dyninfo:dynprog:1} is changed to
	\begin{align}
		&\quad~\hat{q}_t^i(\pi_t, u_t) \\
		&:= \sum_{\tilde{x}_t} r_t^i(\tilde{x}_t, u_t)\pi_t(\tilde{x}_t)  + \begin{cases}
			V_{t+1}^i(\ell_t(\pi_t, u_t))&\text{if }u_t\not\in \overline{\mathcal{U}}_t \\
			0 &\text{if }u_t\in \overline{\mathcal{U}}_t
		\end{cases}
	\end{align}
	for $i\in\{A, B\}$. Then the CBB strategies $(\lambda^{*A}, \lambda^{*B})$ form (the strategy part of) a PBE, and $V_1^A(\hat{\pi})$ and $V_1^B(\hat{\pi})$ are the equilibrium payoff for the principal and the receiver respectively in this PBE.
\end{prop}

\begin{proof}
    Similar to Theorem \ref{thm:dyninfodesign}.
\end{proof}

\section{EXAMPLES}\label{sec:did:examples}
We implement the sequential decomposition algorithm of Proposition \ref{prop:dyninfodesignwquit} in MATLAB for binary state spaces (i.e. $|\mathcal{X}_t| = 2$). We run the algorithm on the following examples of the signal picking game. 

\begin{example}\label{ex:did:farhadi}
	Consider the quickest detection game defined in \cite{farhadi2022dynamic}. In this game, the underlying state $X_t$ is binary and uncontrolled, with $\mathcal{X}_t = \{1, 2\}$. State $2$ is an absorbing state, i.e. $\Pr(X_{t+1} = 2 ~|~ X_{t} = 2) = 1$,
	whereas the system can jump from state $1$ to state $2$ at any time with probability $p$, i.e. $\Pr(X_{t+1} = 2 ~|~ X_{t} = 1) = p$
	where $p\in (0, 1)$.
	
	The receiver would like to detect (the epoch of) the jump from state $1$ to state $2$ as accurately as possible. At each time the receiver has two options: $U_t = j$ stands for declaring state $j$ for $j=1,2$. The instantaneous reward of the receiver is given by
	\begin{align}
		r_t^B(X_t, U_t) = \begin{cases}
			-1&\text{if }X_t = 1, U_t = 2\\
			-c&\text{if }X_t = 2, U_t = 1\\
			0&\text{otherwise}
		\end{cases}
	\end{align}
	where $c\in (0, 1)$. Once the receiver declares state $2$, the game ends immediately.
	
	The principal would like the receiver to stay in the system as long as possible. The instantaneous reward for the principal is
	\begin{align}
		r_t^A(X_t, U_t) = \begin{cases}
			1&\text{if }U_t = 1\\
			0&\text{otherwise}
		\end{cases}
	\end{align}
	
	Setting $p = 0.2, c = 0.1$, we obtained the $q_t^B$ and $V_t^A$ functions specified in Proposition \ref{prop:dyninfodesignwquit} in Figure \ref{fig:did:farhadi1}. The horizontal axis represents $\pi_t(1)$. In the figures for $V_t^A$ functions, the vertices of the triangulation $\gamma_t$ are labeled. The vertices represent the set of beliefs that the principal could induce, and they completely describe the principal's CBB strategy. If the vertex is labeled with red circles, the receiver will take action $U_t=1$ at this posterior belief. If, instead, the vertex is labeled with blue triangles, the receiver will take action $U_t=2$ at this posterior belief.
	
	From the figures, one can see that at any stage, there is only one possible belief that the principal would induce which leads to the receiver quitting the game (i.e. select $U_t=2$). This is consistent with the principal's objective of keeping the receiver in the system. Just like in static information design problems \cite{kamenica2011bayesian,bergemann2019information}, when it is better off for the receiver to declare change, i.e., quit, under the current belief, the principal would promise to tell the receiver that the state is $2$ with some probability $\tilde{p}$ when the state is indeed $2$, and tell the receiver nothing otherwise. In doing so, the receiver would believe that the state is $1$ with a higher probability when the principal does not tell the receiver anything. The principal chooses $\tilde{p}$ to be precisely the value for which the receiver is willing to stay in the system \cite{kamenica2011bayesian}.
	
	When $t$ is close to $T$, the end of the game, the principal would only prefer to declare state $2$ if they believe that $\pi_t(1)$ is very small. This is due to the fact that ``false alarms'' are costlier than delayed detection in this game. When $t$ is further away from $T$, the threshold of $\pi_t(1)$ for the principal to declare state $2$ becomes larger. This holds because when the game is close to end, the receiver has the ``safe'' option to declare state $1$ (at a small cost) until the end to avoid false alarms  (which are costly). However, this option is less preferable when the gap between $t$ and $T$ is large. 
	
	When $t$ is further away from $T$, the principal's value function seems to converge. This is due to the fact that the receiver has the option to quit the game and staying in the game is costly in general.
\end{example}

\begin{example}\label{ex:did:2}
	Consider a game between a principal and a detector. 
	In this game, the underlying state $X_t$ is binary and uncontrolled with $\mathcal{X}_t = \{-1, 1\}$. At any time, the system can jump to the other state with probability $p\in (0,1)$, i.e. 
	\begin{equation}
		\Pr(X_{t+1} = -j ~|~ X_{t} = j) = p,\qquad \forall j\in\{-1, 1\}.
	\end{equation}
	
	The receiver has three actions: $U_t = j$ stands for declaring state $j$ for $j=-1,1$. Both $U_t = -1$ and $U_t = +1$ terminate the game. In addition, the receiver can choose to wait at a cost with action $U_t = 0$. The instantaneous reward of the receiver for $c\in(0,1)$ is given by
	\begin{align}
		r_t^B(X_t, U_t) = \begin{cases}
			1&\text{if }X_t = U_t\\
			-c&\text{if }U_t = 0\\
			0&\text{otherwise}
		\end{cases}.
	\end{align}
 
	The principal would like the receiver to stay in the system as long as possible. The instantaneous reward for the principal is
	\begin{align}
		r_t^A(X_t, U_t) = \begin{cases}
			1&\text{if }U_t = 0\\
			0&\text{otherwise}
		\end{cases}
	\end{align}
	
	Setting $p = 0.2$ and $c = 0.15$, we obtained the $q_t^B$ and $V_t^A$ functions specified in Proposition \ref{prop:dyninfodesignwquit}---see Figure~\ref{fig:did:ex3}. The horizontal axis represents $\pi_t(-1)$. The figures follows the same interpretation as the figures in Example \ref{ex:did:farhadi}. (The markers for actions are different from previous figures, but they are self-explanatory.)
	
	Different from Example \ref{ex:did:farhadi}, the value functions and CBB strategies at equilibrium oscillate with a period of $4$ (given $p = 0.2, c = 0.15$) instead of converging as $t$ gets further away from the horizon $T$.
	
\end{example}

\section{DISCUSSION}\label{sec:did:dis}

Naturally, one may consider extending the above result to two settings: (a) when a public noisy observation of the state is available in addition to the principal's experiment, (b) when there are multiple receivers. However, our result is immediately extendable to neither setting. This is since the techniques we use in this paper depend heavily on the piecewise linear structure of $\hat{q}$ and $V$-functions in \eqref{eq:dyninfo:dynprog}, as well as the preservation of this piecewise linear structure under backward induction. Specifically, when the functions $\hat{q}_t^A, \hat{q}_t^B$ are piecewise linear, the concave closure of $\hat{v}_t^A$ can be expressed as a triangulation based interpolation (through Lemma \ref{lem:did:cavmax}), which in turn allows us to apply the same triangulation to $\hat{v}_t^B$, and thus ensuring the continuity and piecewise linearity of $V_t^B$. However, this structure does not appear in general in the extensions. 

We describe an attempt to extend Theorem \ref{thm:dyninfodesign} to settings (a) and (b) in the most straightforward way. In the case of setting (a), one needs to change the belief update in \eqref{eq:dyninfo:dynprog:1} from $\ell_t(\pi_t, u_t)$ to some other update function that incorporates the public observation. However, unlike $\ell_t(\pi_t, u_t)$, the new update function may not be linear in $\pi_t$. Therefore this procedure cannot preserve piecewise linear properties.

In the case of setting (b), $u_t$ will represent a vector of actions of all receivers, and one needs to change the definition of $\Upsilon_t(\pi_t)$ in \eqref{eq:dyninfo:dynprog:2} to be the set of mixed strategy Nash equilibrium (or alternatively correlated equilibrium) action profiles of the following stage game: Receiver $i$ chooses an action in $\mathcal{U}_t^i$, and receives payoff $\hat{q}_t^i(\pi_t, u_t)$. In this setting, $\Upsilon_t(\pi_t)$ is a set of probability measures on the product set $\mathcal{U}_t$. The new $\hat{v}_t^A$ function can then be given by
$$\hat{v}_t^A(\pi_t)=\underset{\eta_t\in \Upsilon_t(\pi_t)}{\arg\max}~ \sum_{\tilde{u}_t} q_t^A(\pi_t, \tilde{u}_t) \eta_t(\tilde{u}_t)$$

However, in this case, continuity and piecewise linearity of $\hat{q}_t$ are not enough to ensure that the value function $V_t^A$ possesses the same property. To see this, consider the following example with two receivers $B$ and $C$. Let $\mathcal{U}_t^B = \mathcal{U}_t^C = \{1, 2\}=\mathcal{X}_t = \{1, 2\}$. Let $p = \pi_t(1)$. 
Then, all functions of $\pi_t$ can be expressed as a function of $p$. Suppose that
\begin{align*}
	q_t^B(\pi_t, u_t) &= \begin{cases}
		1&\text{if }u_t^B = u_t^C\\
		0&\text{otherwise}
	\end{cases},\\
	q_t^C(\pi_t, u_t) &= \begin{cases}
		p +1 & \text{if }u_t^B = 1, u_t^C = 2\\
		1&\text{if }u_t^B = 2, u_t^C = 1\\
		0&\text{otherwise}
	\end{cases}.
\end{align*}

It can be verified that, under either the concept of Nash equilibrium or correlated equilibrium, $\Upsilon_t(\pi_t)$ contains only one element: player $B$ plays action $1$ with probability $\frac{1}{2+p}$ and player $C$ plays their two actions with equal probability independent of player $B$'s action. Now suppose that
\begin{equation}
	q_t^A(\pi_t, u_t) = \begin{cases}
		p &u_t^B = 1\\
		0&\text{otherwise}
	\end{cases}.
\end{equation}

Then we have $\hat{v}_t^A(\pi_t) = \frac{p}{2+p}$ for $p\in [0, 1]$. Observe that $\hat{v}_t^A$ is a strictly concave function. Hence the concave closure of $\hat{v}_t^A$ is just $\hat{v}_t^A$ itself, which is not piecewise linear.

\section{CONCLUSION AND FUTURE WORK}\label{sec:concl}
In this work, we formulated a dynamic information disclosure game, called the signal picking game, where the principal sequentially commit to a signal/experiment to communicate with the receiver. We showed that there exist equilibria where both the principal and the receiver make decisions based on the canonical belief instead of their respective full information. We also provided a sequential decomposition procedure to find such equilibria. 

Unlike the CIB-belief-based sequential decomposition procedures of \cite{ouyang2016dynamic,tavafoghi2016stochastic,vasal2019spbe,tang2022dynamic}, the sequential decomposition procedure of Theorem \ref{thm:dyninfodesign} always has a solution. The main reason is that the CIB belief in the signal picking game is strategy-independent, just like in \cite{nayyar2013common}. The result illustrates the critical difference between strategy-dependent and strategy-independent compression of information in dynamic games.

There are a few future research problems arising from this work. The first problem is to extend our result to infinite horizon games. The second problem is to extend our result to settings with multiple senders.


\bibliographystyle{IEEEtran}
\bibliography{thesisbib}
\clearpage

\begin{figure}[!htp]
	\centering
	\includegraphics[width=.45\textwidth,keepaspectratio]{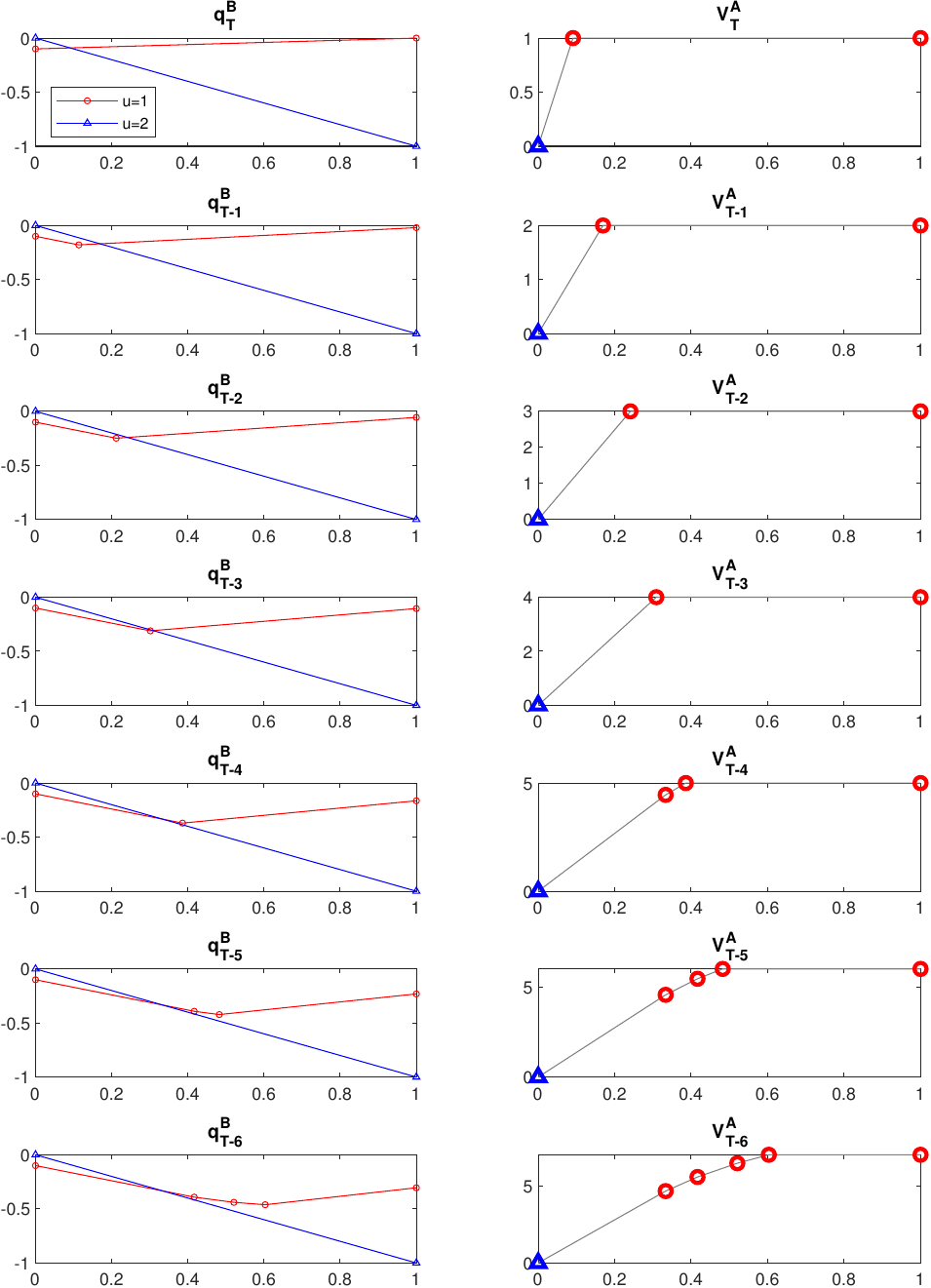}\\
	\includegraphics[width=.45\textwidth,keepaspectratio]{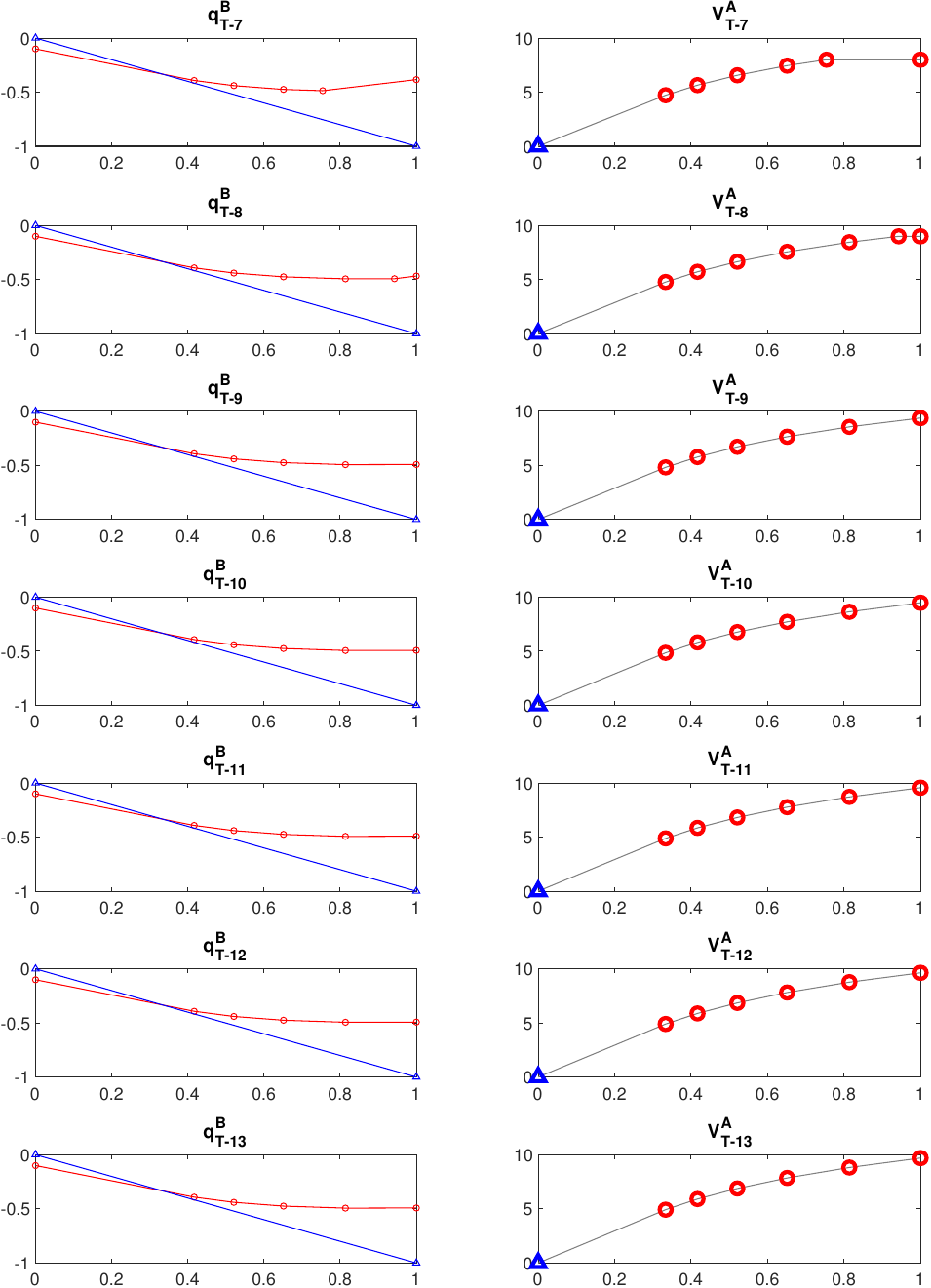}
	\caption{The $q_t^B$ and $V_t^A$ functions for Example \ref{ex:did:farhadi} with $p = 0.2, c = 0.1$ at times $t=T:T-13$.}\label{fig:did:farhadi1}
\end{figure}

\begin{figure}[!htp]
	\centering
	\includegraphics[width=.45\textwidth,keepaspectratio]{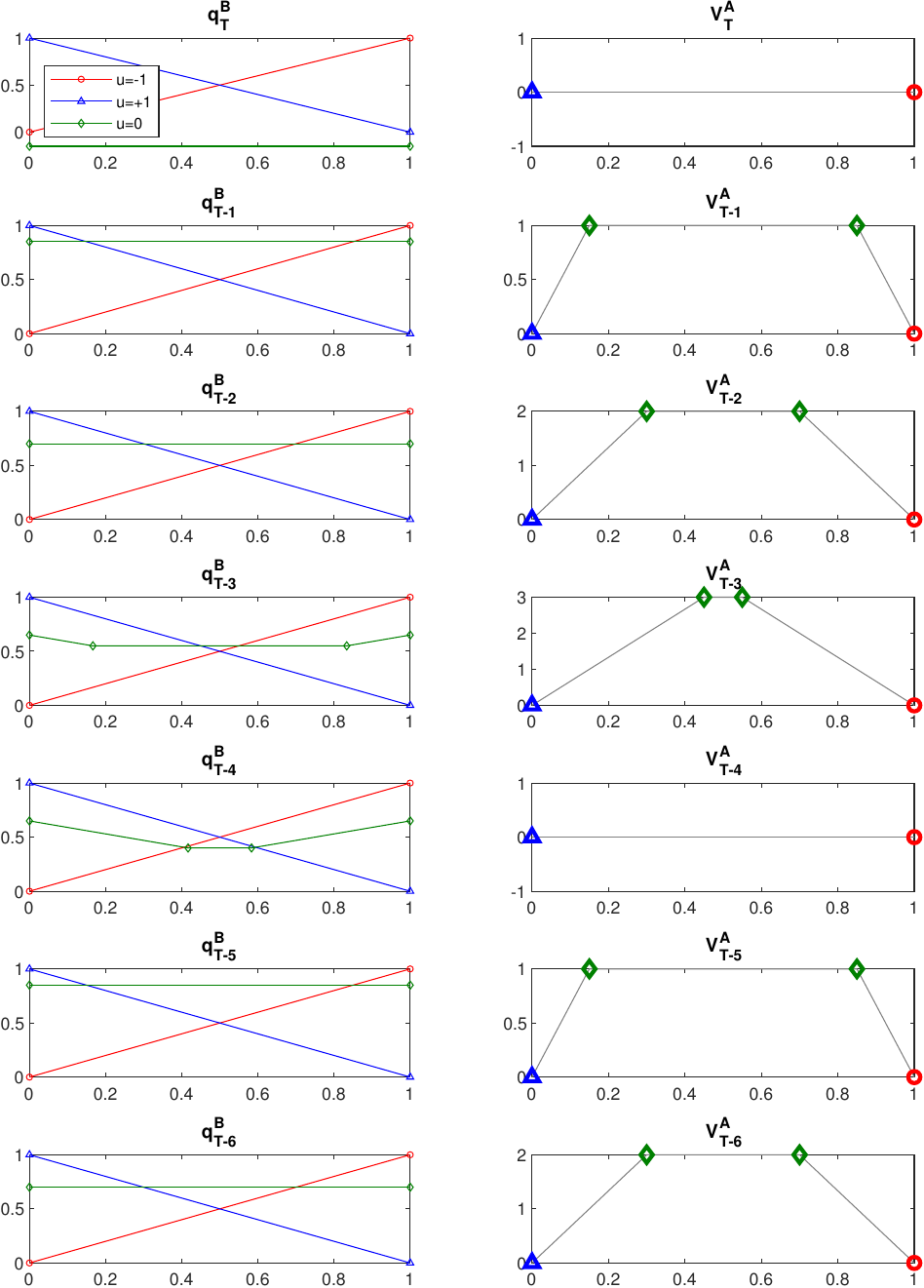}\\
	\includegraphics[width=.45\textwidth,keepaspectratio]{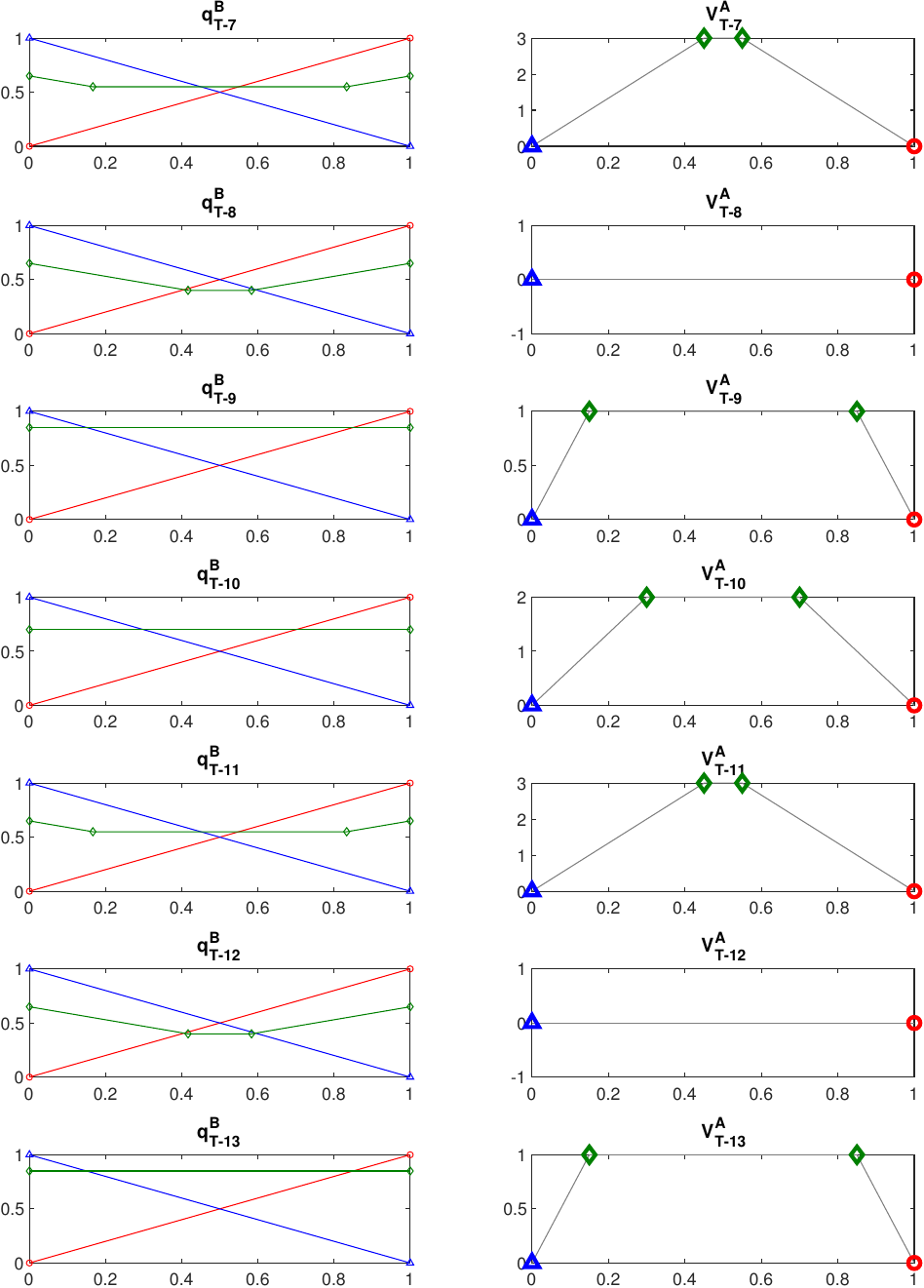}
	\caption{The $q_t^B$ and $V_t^A$ functions for Example \ref{ex:did:2} with $p = 0.2, c = 0.15$ at times $t=T:T-13$.}\label{fig:did:ex3}
\end{figure}

\onecolumn 
\appendix

\subsection{Proofs of Discrete Geometric Results}\label{app:disgeo}
\begin{proof}[Proof of Lemma \ref{lem:did:plcompose}]
	Let $C_1, \cdots, C_k$ be polytopes such that (i) $f$ is linear on each of $C_1, \cdots, C_k$ (ii) $C_1\cup \cdots \cup C_k = \Omega_2$.
	
	Since $\ell$ is an affine function, we have the pre-images $D_j = \ell^{-1}(C_j), j=1,\cdots, k$ to be polytopes as well.
	
	$f\circ \ell$ is linear on each $D_j$ (since it is the composition of two linear functions), and $D_1\cup \cdots \cup D_k = \Omega_1$. We conclude that $f\circ \ell$ is a piecewise linear function.
\end{proof}

\begin{proof}[Proof of Lemma \ref{lem: intpcont}]
	First, for any $\omega$, given a simplex $C$ such that $\omega\in C$, there is a unique way to represent $\omega$ as a convex combination of vertices of $C$.
	
	Suppose that $\omega$ is in both simplices $C$ and $C'$. Then $\omega$ is in $C\cap C'$, which is a face of both $C$ and $C'$. Since $C\cap C'$ is a simplex, $\omega$ can be uniquely represented as a convex combination of vertices of $C\cap C'$. Since the set of vertices of $C\cap C'$ is a subset of the vertices of both $C$ and $C'$, we conclude that the above representation is also the unique way of representing $\omega$ as a convex combination of vertices of $C$ (and of $C'$). We conclude that for any $\omega$, there is a unique way to represent $\omega$ as a convex combination of vertices of any simplex in $\gamma$. Hence $\mathbb{I}(f, \gamma)$ is well defined.
	
	For any simplex $C\in\gamma$, $\mathbb{I}(f, \gamma)$ is linear on $C$. Since the number of simplices in $\gamma$ is finite and their union is $\Omega$, we conclude that $\mathbb{I}(f, \gamma)$ is a continuous piecewise linear function on $\Omega$.
\end{proof}

\begin{proof}[Proof of Lemma \ref{lem:did:cavmax}]
	For $j=1,\cdots, k$, let $C_{j1}, \cdots, C_{jm_j}$ be polytopes corresponding to the piecewise linear function $f_j$ under Definition \ref{def:did:piecewiselinear}, i.e. $f_j$ is linear on each of $C_{j1}, \cdots, C_{jm_j}$ and $C_{j1}\cup \cdots\cup C_{jm_j} = \Omega$. Define
	\begin{equation}
		\mathcal{S} = \{ C_{1 i_1}\cap C_{2 i_2}\cap \cdots \cap C_{k i_k} : 1\leq i_1 \leq m_1  , \cdots, 1\leq i_k\leq m_k \}
	\end{equation}
	
	$\mathcal{S}$ is a finite collection of polytopes. All of $f_1, \cdots, f_k$ are linear on each element of $\mathcal{S}$. The union of $\mathcal{S}$ equals $\Omega$. Define
	\begin{align}
		A_j &:= \{\omega\in\Omega: f_j(\omega) \geq f_{j'}(\omega)\quad\forall j'=1,\cdots, k  \},\\
		\mathcal{S}_j &= \{F \cap A_j : F\in\mathcal{S} \}.
	\end{align}
	
	$\mathcal{S}_j$ is the collection of subsets where $f_j$ is (one of) the maximum among $f_1, \cdots, f_k$. $\mathcal{S}_j$ is also a finite collection of polytopes, since each $F \cap A_j$ is a subset of $F$ that satisfy certain linear constraints.
	
	Similarly, let $D_{j1}, \cdots, D_{jn_j}$ be polytopes corresponding to the piecewise linear function $\rho_j$. Define
	\begin{equation}
		\mathcal{R}_j = \{F \cap D_{ji}: F\in \mathcal{S}_j, 1\leq i\leq n_j\}
	\end{equation}
	
	For each polytope $F \in \mathcal{R}_j$, $\rho_j$ is linear on $F$, and $f_j$ is (one of) the maximum among $f_1, \cdots, f_k$ for all points in $F$. The union of $\mathcal{R}_j$ equals $A_j$.
	
	Let $\mathcal{P}$ be the set of vertices of polytopes in $\mathcal{R}_1 \cup \cdots \cup \mathcal{R}_k$. $\mathcal{P}$ is a finite set. Define $\mathcal{B}\subset \Omega \times \mathbb{R}$ by $ \mathcal{B} = \{(\omega, \Psi(\omega)): \omega \in\mathcal{P} \}$. Let $\mathcal{Z}$ be the convex hull of $\mathcal{B}$. We have $\mathcal{Z}$ to be a polytope with its vertices contained in $\mathcal{B}$. 
	
	Let $\hat{\Psi}$ be the concave closure of $\Psi$.
	We will show that the function $\hat{\Psi}$ is represented by the upper face of $\mathcal{Z}$. Then we obtain a triangulation of $\Omega$ by projecting a triangulation of $\mathcal{Z}$ in a similar way to the construction of regular triangulations (see Section 2.2 of \cite{de2010triangulations}).
	
	\textbf{Step 1:} Prove that $\hat{\Psi}(\omega) = \max \{y: (\omega, y)\in \mathcal{Z}\}$.
	
	Define $\overline{\Psi}(\omega) := \max \{y: (\omega, y)\in \mathcal{Z}\}$. $\overline{\Psi}$ is a concave function.
	
	It is clear that $\mathcal{Z} \subset \mathrm{cvxg}(\Psi)$, since $\mathcal{B}$ is a subset of the graph of $\Psi$. Therefore $\overline{\Psi}(\omega)\leq \hat{\Psi}(\omega)$.
	
	Consider any $\omega\in\Omega$. Let $j^*$ be such that $j^* \in \Upsilon(\omega)$ and $\Psi(\omega) = \rho_{j^*}(\omega)$. Then $\omega\in A_{j^*}$. We have $\omega\in F$ for some $F \in \mathcal{R}_{j^*}$. Let $\omega_1, \cdots, \omega_m \in \mathcal{P}$ be the vertices of $F$. We can write $\omega=\alpha_1 \omega_1 + \cdots + \alpha_k \omega_m$ for some $\alpha_1, \cdots, \alpha_m\geq 0, \alpha_1 + \cdots + \alpha_m = 1$. 	
	Since $\rho_{j^*}$ is linear on $F$ we have
	\begin{equation}\label{eq:did:hsg}
		\Psi(\omega) = \rho_{j^*}(\omega) = \alpha_1 \rho_{j^*}(\omega_1)+ \cdots + \alpha_k \rho_{j^*}(\omega_m)
	\end{equation}
	
	Since $\omega_1, \cdots, \omega_m \in F$ and $F\subset A_{j^*} $. By definition, $j^*\in \Upsilon(\omega_i)$ for all $i=1,\cdots,m$. Therefore $\rho_{j^*}(\omega_i) \leq \Psi(\omega_i)$ for all $i=1,\cdots,m$. Consequently, combining \eqref{eq:did:hsg} we have
	\begin{equation}
		\Psi(\omega) \leq \alpha_1 \Psi(\omega_1) + \cdots + \alpha_k \Psi(\omega_m)
	\end{equation}
	
	Given that $(\omega, \alpha_1 \Psi(\omega_1) + \cdots + \alpha_k \Psi(\omega_k)) \in \mathcal{Z}$, we have
	\begin{equation}
		\alpha_1 \Psi(\omega_1) + \cdots + \alpha_k \Psi(\omega_k) \leq \overline{\Psi}(\omega)
	\end{equation} 
	
	Hence $\Psi(\omega)\leq \overline{\Psi}(\omega)$. Therefore, $\overline{\Psi}$ is a concave function above $h$. Since the concave closure $\hat{\Psi}(\omega)$ is the smallest concave function above $h$, we conclude that $\hat{\Psi}(\omega)\leq \overline{\Psi}(\omega)$ for all $\omega\in\Omega$.
	
	Therefore $\hat{\Psi} = \overline{\Psi}$, completing the proof of Step 1.
	
	\textbf{Step 2:} Construct the triangulation $\gamma$ and show that $\hat{\Psi} = \mathbb{I}(h, \gamma)$.
	
	Let $\mathcal{A}\subset \Omega \times \mathbb{R}$ be the graph of $\hat{\Psi}$. $\mathcal{A}$ is also the union of upper faces of $\mathcal{Z}$. Let $\vartheta$ be a point set triangulation (as defined in Def. 2.2.1 in \cite{de2010triangulations}) of the finite point set $\mathcal{B}$. (A point set triangulation of a finite set of points always exists. See Section 2.2.1 of \cite{de2010triangulations}.) Let $\hat{\vartheta}$ be the restriction of $\vartheta$ to $\mathcal{A}$, i.e. $\hat{\vartheta}:= \{F: F\subset \mathcal{A}, F\in \vartheta \}$. It can be shown that $\hat{\vartheta}$ is a simplicial complex (i.e. a polyhedral complex where all polytopes are simplices. See Def. 2.1.5 in \cite{de2010triangulations}.) with vertices contained in $\mathcal{A} \cap \mathcal{B}$.
	
	Since $\mathcal{A}$ is the upper convex hull of $\mathcal{Z}$, the projection map $\mathrm{proj}_{\Omega}: \Omega\times \mathbb{R} \mapsto \Omega,  (\omega, y)\mapsto \omega$ is a bijection between $\mathcal{A}$ and $\Omega$. Let $\gamma$ be the projection of $\hat{\vartheta}$ on to $\Omega$, i.e. $\gamma = \{\mathrm{proj}_{\Omega}(F): F\in \hat{\vartheta} \}$. We conclude that $\gamma$ is a simplical complex that spans $\Omega$, i.e. a triangulation of $\Omega$. 
	
	The inverse map $\mathrm{proj}_{\Omega}^{-1}: \Omega \mapsto \mathcal{A}$ is a piecewise linear map that is linear on each simplex in $\gamma$. Therefore we have $\hat{\Psi}(\omega) = \mathbb{I}(\hat{\Psi}, \gamma)(\omega)$ for all $\omega\in\Omega$. Since the vertices of $\hat{\vartheta}$ are contained in both $\mathcal{A}$ and $\mathcal{B}$, we have $\hat{\Psi}(\omega) = \Psi(\omega)$ for each vertex $\omega$ of $\gamma$. (Recall that $\mathcal{B}$ is a subset of the graph of $h$ and $\mathcal{A}$ is the graph of $\hat{\Psi}$.) Therefore we have $\hat{\Psi}(\omega) = \mathbb{I}(\Psi, \gamma)(\omega)$ for all $\omega\in\Omega$.
\end{proof}

\subsection{Proof of Main Results}\label{app:main}

\begin{lemma}\label{lem:canonicalbelief}
	There exist a belief system $\mu^* = (\mu^{*A}, \mu^{*B})$ such that (i) $\mu^*$ is consistent with any strategy profile $(g^A, g^B)$; (ii) the canonical belief system $\kappa$ is the marginals of $\mu^*$.
\end{lemma}

\begin{proof}[Proof of Lemma \ref{lem:canonicalbelief}]
	Define $\mu_t^{*i}: \mathcal{H}_t^i\mapsto \Delta(\mathcal{X}_{1:t})$ recursively through the following:
	\begin{itemize}
		\item $\mu_1^{*A}(h_1^A) := \hat{\pi}$.
		\item $\mu_t^{*B}(x_{1:t}|h_t^B) := \dfrac{\mu_t^{*A}(x_{1:t}|h_t^A)\sigma_t(m_t|x_t) }{\sum_{\tilde{x}_{1:t}} \mu_t^{*A}(\tilde{x}_{1:t}|h_t^A)\sigma_t(m_t|\tilde{x}_t) }$
		\item $\mu_{t+1}^{*A}(x_{1:t+1}|h_{t+1}^A) := \mu_t^{*B}(x_{1:t}|h_t^B) \Pr(x_{t+1}|x_t, u_t)$
	\end{itemize}
	
	Through induction on $t$ it is clear that $\kappa$ is the marginal distribution derived from $\mu^*$. 
	
	It remains to show the consistency of $\mu^*$ w.r.t. any strategy profile $g=(g^A, g^B)$. We will also show it via induction:
	\begin{itemize}
		\item $\mu_1^{*A}$ is clearly consistent with any $g$ since it is defined to be the prior distribution of $X_1$.
		
		\item Suppose that $\mu_t^{*A}$ is consistent with $g$. Then consider any $h_t^B=(\sigma_{1:t}, m_{1:t}, u_{1:t-1})\in\mathcal{H}_t^B$ such that $\Pr^{g}(h_t^B) > 0$. Then we have $\Pr^{g}(h_t^A) > 0$, and $\mu_t^{*A}(x_{1:t}|h_t^A) = \Pr^g (x_{1:t}|h_t^A)$ follows by induction hypothesis. Therefore
		\begin{align*}
			&\quad~\Pr^g(x_{1:t}|h_t^B) = \Pr^g(x_{1:t}|\sigma_t, m_t, h_t^A) = \dfrac{\Pr^{g}(x_{1:t}, \sigma_t, m_t, h_t^A)}{\Pr^{g}(\sigma_t, m_t, h_t^A)}\\
			&=\dfrac{\Pr^{g}(x_{1:t}, h_t^A) \Pr^{g}(\sigma_t|x_{1:t}, h_t^A) \Pr^{g}(m_t|\sigma_t, x_{1:t}, h_t^A) }{\sum_{\tilde{x}_{1:t}} \Pr^{g}(\tilde{x}_{1:t}, h_t^A) \Pr^{g}(\sigma_t|\tilde{x}_{1:t}, h_t^A) \Pr^{g}(m_t|\sigma_t, \tilde{x}_{1:t}, h_t^A)}\\
			&=\dfrac{ \Pr^{g}(x_{1:t}, h_t^A) g_t^A(\sigma_t|h_t^A) \sigma_t(m_t|x_t) }{ \sum_{\tilde{x}_{1:t}} \Pr^{g}(\tilde{x}_{1:t}, h_t^A) g_t^A(\sigma_t|h_t^A) \sigma_t(m_t|\tilde{x}_{t})}\\
			&=\dfrac{ \Pr^{g}(x_{1:t}, h_t^A) \sigma_t(m_t|x_t) }{ \sum_{\tilde{x}_{1:t}} \Pr^{g}(\tilde{x}_{1:t}, h_t^A)  \sigma_t(m_t|\tilde{x}_{t})}=\dfrac{ \Pr^{g}(x_{1:t}| h_t^A) \sigma_t(m_t|x_t) }{ \sum_{\tilde{x}_{1:t}} \Pr^{g}(\tilde{x}_{1:t}| h_t^A)  \sigma_t(m_t|\tilde{x}_{t})}\\
			&=\dfrac{ \mu_t^{*A}(x_{1:t}|h_t^A) \sigma_t(m_t|x_t) }{ \sum_{\tilde{x}_{1:t}} \mu_t^{*A}(\tilde{x}_{1:t}| h_t^A)  \sigma_t(m_t|\tilde{x}_{t})} = \mu_t^{*B}(x_{1:t}|h_t^B)
		\end{align*}
		which means that $\mu_t^{*B}$ is consistent with $g$.
		
		\item Suppose that $\mu_t^{*B}$ is consistent with $g$. Then consider any $h_{t+1}^A = (\sigma_{1:t}, m_{1:t}, u_{1:t})\in \mathcal{H}_{t+1}^A$ such that $\Pr^g(h_{t+1}^A) > 0$. Then we have $\Pr^{g}(h_t^B) > 0$, and $\mu_t^{*B}(x_{1:t}|h_t^B) = \Pr^g (x_{1:t}|h_t^B)$ follows by induction hypothesis. Then we have
		\begin{align*}
			&\quad~\Pr^g(x_{1:t+1}|h_{t+1}^A) = \dfrac{\Pr^g(x_{1:t+1}, h_{t+1}^A)}{\Pr^g(h_{t+1}^A)}\\
			&=\dfrac{\Pr^g(x_{1:t}, h_t^B) \Pr^g(u_t|x_{1:t}, h_t^B) \Pr^g(x_{t+1}| x_{1:t}, u_t, h_t^B)}{\sum_{\tilde{x}_{1:t}} \Pr^g(\tilde{x}_{1:t}, h_t^B) \Pr^g(u_t|\tilde{x}_{1:t}, h_t^B) }\\
			&=\dfrac{\Pr^g(x_{1:t}, h_t^B) g_t^B(u_t| h_t^B) \Pr(x_{t+1}| x_{t}, u_t)}{\sum_{\tilde{x}_{1:t}} \Pr^g(\tilde{x}_{1:t}, h_t^B) g_t(u_t| h_t^B) }\\
			&=\dfrac{\Pr^g(x_{1:t}, h_t^B) }{\sum_{\tilde{x}_{1:t}} \Pr^g(\tilde{x}_{1:t}, h_t^B) } \cdot \Pr(x_{t+1}| x_{t}, u_t) = \mu_t^{*B}(x_{1:t}|h_t^B) \Pr(x_{t+1}| x_{t}, u_t)\\
			&=\mu_t^{*A}(x_{1:t+1}|h_{t+1}^A)
		\end{align*}
		which means that $\mu_{t+1}^{*A}$ is consistent with $g$.
	\end{itemize}
\end{proof}

\begin{proof}[Proof of Theorem \ref{thm:dyninfodesign}]
	Let $\mu^*$ be a belief system that satisfies Lemma \ref{lem:canonicalbelief}. It is shown by Lemma \ref{lem:canonicalbelief} that $\mu^*$ is consistent with any strategy profile $g$. Hence to show that $\lambda^*$ forms a CBB-PBE we only need to show sequential rationality.
	
	\textbf{Step 1:} Fixing the principal's strategy to be $\lambda^{*A}$, show that $\lambda_{\tau: T}^{*B}$ is sequentially rational at any $h_\tau^B\in \mathcal{H}_\tau^B$ at any time $\tau$, given the belief $\mu_\tau^{*B}(h_\tau^B)$.
	
	To prove Step 1, we argue that at $h_{\tau}^B$, the receiver is facing an MDP problem with state process $\Pi_t^B = \kappa_t^B(H_t^B)$ and action $U_t$ for $t\geq \tau$.
	
	First, we can write
	\begin{align*}
		&\quad~\E^{\mu_{\tau}^{*B}(h_\tau^B), \lambda_t^{*A}, g_{\tau:T}^B}\left[\sum_{t=\tau}^T r_t^B(X_t, U_t)\right]\\
		&=\E^{\mu_{\tau}^{*B}(h_\tau^B), \lambda_t^{*A}, g_{\tau:T}^B}\left[\sum_{t=\tau}^T \E^{\mu_{\tau}^{*B}(h_\tau^B), \lambda_t^{*A}, g_{\tau:T}^B}[r_t^B(X_t, U_t)| H_t^B, U_t]\right]
	\end{align*}
	where for any $h_t^B$ such that $\Pr^{\mu_{\tau}^{*B}(h_\tau^B), \lambda_t^{*A}, g_{\tau:T}^B}(h_t^B) > 0$ we have
	\begin{align*}\E^{\mu_{\tau}^{*B}(h_\tau^B), \lambda_t^{*A}, g_{\tau:T}^B}[r_t^B(X_t, U_t)| h_t^B, u_t]&=\sum_{\tilde{x}_{1:t}} r_t^B(\tilde{x}_t, u_t) \Pr^{\mu_{\tau}^{*B}(h_\tau^B), \lambda_t^{*A}, g_{\tau:T}^B}(\tilde{x}_{1:t}|h_t^B)\\
		&=\sum_{\tilde{x}_t} r_t^B(\tilde{x}_t, u_t)\pi_t^{B}(\tilde{x}_t)=:\tilde{r}_t^B(\pi_t^B, u_t)
	\end{align*}
	where $\pi_t^B:=\kappa_t^{B}(h_t^B)$. The second equality is true due to Lemma \ref{lem:canonicalbelief}.
	
	Therefore we can write
	\begin{align*}
		&\quad~\E^{\mu_{\tau}^{*B}(h_\tau^B), \lambda_t^{*A}, g_{\tau:T}^B}\left[\sum_{t=\tau}^T r_t^B(X_t, U_t)\right]=\E^{\mu_{\tau}^{*B}(h_\tau^B), \lambda_t^{*A}, g_{\tau:T}^B}\left[\sum_{t=\tau}^T \tilde{r}_t^B(\Pi_t^B, U_t) \right].
	\end{align*}
	
	We now show that $\Pi_{t}^B$ is a controlled Markov Chain controlled by $U_t$. By Definition \ref{def:canonicalbelief}, we have $\Pi_{t+1}^B =\xi_{t+1} (\Pi_{t+1}^A, \Sigma_{t+1}, M_{t+1})$, where $\Pi_{t+1}^A=\ell_t(\Pi_t^B, U_t), \Sigma_{t+1} = \lambda_{t+1}^{*B}(\Pi_{t+1}^A)$. Therefore we have
	\begin{align*}
		&\quad~ \Pr^{\mu_{\tau}^{*B}(h_\tau^B), \lambda_t^{*A}, g_{\tau:T}^B}(\pi_{t+1}^B|h_t^B, u_t)\\
		&=\sum_{\tilde{m}_{t+1}}  \bm{1}_{\{\pi_{t+1}^B = \xi_{t+1}(\pi_{t+1}^A, \sigma_{t+1}, \tilde{m}_{t+1}) \} } \Pr^{\mu_{\tau}^{*B}(h_\tau^B), \lambda_t^{*A}, g_{\tau:T}^B}(\tilde{m}_{t+1}|h_t^B, u_t)\\
		&=\sum_{\tilde{m}_{t+1}}  \bm{1}_{\{\pi_{t+1}^B = \xi_{t+1}(\pi_{t+1}^A, \sigma_{t+1}, \tilde{m}_{t+1}) \} } \sum_{\tilde{x}_{t+1}}\sigma_{t+1}(\tilde{m}_{t+1}|\tilde{x}_{t+1}) \Pr^{\mu_{\tau}^{*B}(h_\tau^B), \lambda_t^{*A}, g_{\tau:T}^B}(\tilde{x}_{t+1}|h_{t+1}^A)\\
		&=\sum_{\tilde{m}_{t+1}}  \bm{1}_{\{\pi_{t+1}^B = \xi_{t+1}(\pi_{t+1}^A, \sigma_{t+1}, \tilde{m}_{t+1}) \} } \sum_{\tilde{x}_{t+1}}\sigma_{t+1}(\tilde{m}_{t+1}|\tilde{x}_{t+1}) \pi_{t+1}^A(\tilde{x}_{t+1})
	\end{align*}
	where $\pi_{t+1}^A = \ell_t(\pi_t^B, u_t), \sigma_{t+1} = \lambda_{t+1}^{*A}(\pi_{t+1}^A)$. The last equality is true due to Lemma \ref{lem:canonicalbelief}.
	
	By construction, $\sigma_{t+1}= \lambda_{t+1}^{*A}(\pi_{t+1}^A)$ induces the distribution $\mathbb{C}(\pi_{t+1}^A, \gamma_{t+1})$ from $\pi_{t+1}^A$. This means that
	\begin{align*}
		&\quad~\sum_{\tilde{m}_{t+1}}  \bm{1}_{\{\pi_{t+1}^B = \xi_{t+1}(\pi_{t+1}^A, \sigma_{t+1}, \tilde{m}_{t+1}) \} } \sum_{\tilde{x}_{t+1}}\sigma_{t+1}(\tilde{m}_{t+1}|\tilde{x}_{t+1}) \pi_{t+1}^A(\tilde{x}_{t+1}) \\
		&= \mathbb{C}(\pi_{t+1}^A, \gamma_{t+1})(\pi_{t+1}^B)
	\end{align*}
	
	We conclude that
	\begin{align*}
		\Pr^{\mu_{\tau}^{*B}(h_\tau^B), \lambda_t^{*A}, g_{\tau:T}^B}(\pi_{t+1}^B|h_t^B, u_t)&= \mathbb{C}(\ell_t(\pi_{t}^B, u_t), \gamma_{t+1})(\pi_{t+1}^B).
	\end{align*}
	
	In particular, this means that the conditional distribution of $\Pi_{t+1}^B$ given all of $(\Pi_{1:t}^B, U_{1:t}^B)$ is dependent only on $(\Pi_{t}^B, U_t)$, proving that $\Pi_{t}^B$ is a controlled Markov Chain controlled by $U_t$ for $t\geq \tau$. 
	
	Therefore, at $h_\tau^B$, the receiver faces an MDP problem with state $\Pi_t^B$, action $U_t$, instantaneous reward $\tilde{r}_t^B(\Pi_t^B, U_t)$ and transition kernel $\Pr(\pi_{t+1}^B| \pi_t^B, u_t) = \mathbb{C}(\ell_t(\pi_{t}^B, u_t), \gamma_{t+1})(\pi_{t+1}^B)$.
	
	Now, by construction, we know that
	\begin{align}
		\hat{v}_t^{B}(\pi_t^B) &= \max_{u_t} \left[\sum_{\tilde{x}_t} r_t^B(\tilde{x}_t, u_t)\pi_t(\tilde{x}_t)  + V_{t+1}^B(\ell_t(\pi_t, u_t))\right] \\
		&= \max_{u_t} \left[\tilde{r}_t^B(\pi_t^B, u_t) + \int \hat{v}_{t+1}^B(\cdot) \mathrm{d} \mathbb{C}(\ell_t(\pi_t^B, u_t) , \gamma_{t+1}) \right]\label{eq:dyninfo:bellman}
	\end{align}
	and $\lambda_t^{*B}(\pi_t^B)$ attains the maximum in \eqref{eq:dyninfo:bellman}. Therefore $\lambda_{\tau:T}^{*B}$ solves the Bellman equation for the MDP problem specified above, and hence is an optimal strategy. Furthermore, $V_1^B(\hat{\pi}) = \int \hat{v}_1^B(\cdot) \mathrm{d} \mathbb{C}(\hat{\pi}, \gamma_t)$ is the optimal total expected payoff for the receiver when the principal plays $\lambda^{*A}$.
	
	\textbf{Step 2:} Fixing the receiver's strategy to be $\lambda^{*B}$, show that $\lambda_{\tau: T}^{*A}$ is sequentially rational at any $h_\tau^A\in \mathcal{H}_\tau^A$ at any time $\tau$, given the belief $\mu_\tau^{*A}(h_\tau^A)$.
	
	Similar to Step 1, we argue that at $h_\tau^A$, the principal is facing an MDP problem with state process $\Pi_t^A=\kappa_t^A(H_t^A)$ and action $\Sigma_t$ for $t\geq \tau$.
	
	First, we write
	\begin{align*}
		\E^{\mu_{\tau}^{*A}(h_\tau^A), g_{\tau:T}^A, \lambda_t^{*B}}\left[\sum_{t=\tau}^T r_t^A(X_t, U_t)\right]&=\E^{\mu_{\tau}^{*A}(h_\tau^A), g_{\tau:T}^A, \lambda_t^{*B}}\left[\sum_{t=\tau}^T \E^{\mu_{\tau}^{*A}(h_\tau^A), g_{\tau:T}^A, \lambda_t^{*B}}[r_t^A(X_t, U_t)| H_t^A, \Sigma_t]\right]
	\end{align*}
	
	Given that the receiver uses the CBB strategy $\lambda^{*B}$, we know that $U_t=\lambda_t^{*B}(\Pi_t^B)$ where $\Pi_t^B=\xi_t(\Pi_{t}^A, \Sigma_t, M_t)$. For any $h_t^A$ such that $\Pr^{\mu_{\tau}^{*A}(h_\tau^A), g_{\tau:T}^A, \lambda_t^{*B}}(h_t^A) > 0$ we have
	\begin{align*}
		&\quad~\E^{\mu_{\tau}^{*A}(h_\tau^A), g_{\tau:T}^A, \lambda_t^{*B}}[r_t^A(X_t, U_t)| h_t^A, \sigma_t]\\
		&=\sum_{\tilde{x}_{t}, \tilde{m}_t} r_t^A(\tilde{x}_t, \lambda_t^{*B}(\xi_t(\pi_t^A, \sigma_t, \tilde{m}_t)) ) \Pr^{\mu_{\tau}^{*A}(h_\tau^A), g_{\tau:T}^A, \lambda_t^{*B}}(\tilde{m}_t, \tilde{x}_t|h_t^A, \sigma_t) \\
		&=\sum_{\tilde{x}_{t}, \tilde{m}_t} r_t^A(\tilde{x}_t, \lambda_t^{*B}(\xi_t(\pi_t^A, \sigma_t, \tilde{m}_t)) ) \sigma_{t}(\tilde{m}_t|\tilde{x}_t) \Pr^{\mu_{\tau}^{*A}(h_\tau^A), g_{\tau:T}^A, \lambda_t^{*B}}(\tilde{x}_t|h_t^A, \sigma_t)\\
		&=\sum_{\tilde{x}_{t}, \tilde{m}_t} r_t^A(\tilde{x}_t, \lambda_t^{*B}(\xi_t(\pi_t^A, \sigma_t, \tilde{m}_t)) ) \sigma_{t}(\tilde{m}_t|\tilde{x}_t) \pi_t^A(\tilde{x}_t)=:\tilde{r}_t^A(\pi_t^A, \sigma_t)
	\end{align*}
	where $\pi_t^A:=\kappa_t^{A}(h_t^A)$. The third equality is true due to Lemma \ref{lem:canonicalbelief}.
	
	Therefore we can write
	\begin{align*}
		&\quad~\E^{\mu_{\tau}^{*A}(h_\tau^A), g_{\tau:T}^A, \lambda_t^{*B}}\left[\sum_{t=\tau}^T r_t^A(X_t, U_t)\right]=\E^{\mu_{\tau}^{*A}(h_\tau^A), g_{\tau:T}^A, \lambda_t^{*B}}\left[\sum_{t=\tau}^T \tilde{r}_t^A(\Pi_t^A, \Sigma_t) \right].
	\end{align*}
	
	We now show that $\Pi_t^A$ is a controlled Markov process with action $\Sigma_t$: We know that 
	\begin{align*}
		\Pi_{t+1}^{A} = \ell_t(\Pi_t^B, U_t^B), \qquad U_t^B &= \lambda_t^{*B}(\Pi_t^B), \qquad \Pi_t^B = \xi_t(\Pi_{t}^A, \Sigma_t, M_t).
	\end{align*}
	Hence $\Pi_{t+1}^A$ is a function of $\Pi_t^A, \Sigma_t,$ and $M_t$. Furthermore, 
	\begin{align*}
		\Pr^{\mu_{\tau}^{*A}(h_\tau^A), g_{\tau:T}^A, \lambda_t^{*B}}(m_t|h_t^A, \sigma_t)&=\sum_{\tilde{x}_t} \sigma_t(m_t|\tilde{x}_t) \Pr^{\mu_{\tau}^{*A}(h_\tau^A), g_{\tau:T}^A, \lambda_t^{*B}}(\tilde{x}_t|h_t^A, \sigma_t)\\
		&=\sum_{\tilde{x}_t} \sigma_t(m_t|\tilde{x}_t) \pi_t^A(\tilde{x}_t).
	\end{align*}
	Therefore the conditional distribution of $\Pi_{t+1}^A$ given $(H_t^A, \Sigma_t)$ depends only on $(\Pi_t^A, \Sigma_t)$, proving that $\Pi_t^A$ is a controlled Markov process. We conclude that at $h_\tau^A$, the principal faces an MDP problem with state $\Pi_t^A$, action $\Sigma_t$, and instantaneous reward $\tilde{r}_t^A(\Pi_t^A, \Sigma_t)$ for $t\geq \tau$.
	
	Next we will show that $\lambda_{\tau:T}^{*A}$ is a dynamic programming solution of this MDP.
	
	\textbf{Induction Variant:} $V_{t}^A$, as defined in \eqref{eq:dyninfo:dynprog:5}, is the value function for this MDP.
	
	\textbf{Induction Step:} Suppose that $V_{t+1}^A$ is the value function for this MDP at time $t+1$ and consider the stage optimization problem at $\pi_t^A$.

	Note that the instantaneous cost can be written as
	\begin{align*}
		\tilde{r}_t^A(\pi_t^A, \sigma_t)&=\sum_{\tilde{x}_{t}, \tilde{m}_t} r_t^A(\tilde{x}_t, \lambda_t^{*B}(\xi_t(\pi_t^A, \sigma_t, \tilde{m}_t)) ) \sigma_{t}(\tilde{m}_t|\tilde{x}_t) \pi_t^A(\tilde{x}_t)\\
		&=\sum_{\tilde{m}_t}\left(\sum_{\tilde{x}_{t}} r_t^A(\tilde{x}_t, \lambda_t^{*B}(\xi_t(\pi_t^A, \sigma_t, \tilde{m}_t)) ) \dfrac{\sigma_{t}(\tilde{m}_t|\tilde{x}_t) \pi_t^A(\tilde{x}_t)}{\sum_{\hat{x}_t }  \sigma_{t}(\tilde{m}_t|\hat{x}_t) \pi_t^A(\hat{x}_t) } \right)\left(\sum_{\hat{x}_t }  \sigma_{t}(\tilde{m}_t|\hat{x}_t) \pi_t^A(\hat{x}_t)\right) \\
		&=\sum_{\tilde{m}_t}\left(\sum_{\tilde{x}_{t}} r_t^A(\tilde{x}_t, \lambda_t^{*B}(\xi_t(\pi_t^A, \sigma_t, \tilde{m}_t)) ) \xi_t(\pi_t^A, \sigma_t, \tilde{m}_t)(\tilde{x}_t) \right)\left(\sum_{\hat{x}_t }  \sigma_{t}(\tilde{m}_t|\hat{x}_t) \pi_t^A(\hat{x}_t)\right)\\
		&=\E\left[\sum_{\tilde{x}_t} r_t^A(\tilde{x}_t, \lambda_t^{*B}(\Pi_t^B)) \Pi_t^B(\tilde{x}_t)\Big| \pi_t^A, \sigma_t\right]
	\end{align*}
	where $\Pi_t^B$ is a random distribution that follows the distribution induced by $\sigma_t$ from $\pi_t^A$.
	
	Hence objective function for the stage optimization can be written as
	\begin{align*}
		\tilde{Q}_t^A(\pi_t^A, \sigma_{t})&= \tilde{r}_t^A(\pi_t^A, \sigma_{t}) + \E[V_{t+1}^A(\Pi_{t+1}^A)|\pi_t^A, \sigma_t]\\
		&=\tilde{r}_t^A(\pi_t^A, \sigma_{t}) + \E[V_{t+1}^A(\ell_t(\Pi_{t}^B, \lambda_t^{*B}(\Pi_t^B) ))|\pi_t^A, \sigma_t]\\
		&=\E\left[\tilde{v}_t^A(\Pi_t^B) | \pi_t^A, \sigma_t \right]
	\end{align*}
	where
	\begin{align*}
		\tilde{v}_t^A(\pi_t) := \sum_{\tilde{x}_t} r_t^A(\tilde{x}_t, \lambda_t^{*B}(\pi_t)) \pi_t(\tilde{x}_t) + V_{t+1}^A(\ell_t(\pi_{t}^B, \lambda_t^{*B}(\pi_t^B)))
	\end{align*}
	
	By construction of $\lambda_t^{*B}$, we know that $\tilde{v}_t^A = \hat{v}_t^A$ (defined in \eqref{eq:dyninfo:dynprog:vta}).
	
	Therefore, the stage optimization problem can be reformulated as:
	\begin{equation}\label{eq:dyninfo:stageprob}
		\begin{split}
			\max_{\nu_t\in \Delta_f (\Delta(\mathcal{X}_t)) }&\quad \int \hat{v}_t^A(\cdot)\mathrm{d} \nu_t\\
			\text{subject to}&\quad\nu_t\text{ is inducible from $\pi_t^A$}
		\end{split}\tag{SP}
	\end{equation}
	and the optimal signal is any signal that induces an optimal distribution $\nu_t^*$ of \eqref{eq:dyninfo:stageprob} from $\pi_t^A$.
	
	By the seminal result on one-shot information design in \cite{kamenica2011bayesian}, we know that the optimal value of \eqref{eq:dyninfo:stageprob} is given by the concave closure of the function $\hat{v}_t^A$ evaluated at $\pi_t^A$. 
	
	By construction, we have $V_t^A$ to be the concave closure of $\hat{v}_t^A$. Furthermore, $\lambda_t^{*A}(\pi_t^A)$ is assumed to induce the distribution $\nu_t=\mathbb{C}(\pi_t^A, \gamma_t)$, where we know that $\int \hat{v}_t^A(\cdot)\mathrm{d} \mathbb{C}(\pi_t^A, \gamma_t) = V_t^A(\pi_t^A)$. Hence $\lambda_t^{*A}(\pi_t^A)$ is an optimal solution for the stage optimization problem, and $V_t^A$ is the value function at time $t$, proving the induction step.\\
	
	We conclude that $\lambda_{\tau:T}^{*A}$ is an optimal strategy for the principal at $h_\tau^A$ given the belief system $\mu_\tau(h_\tau^A)$ and the receiver's strategy $\lambda^{*B}$. Furthermore, $V_1^A(\hat{\pi})$ is the optimal total expected payoff for the principal when the receiver plays $\lambda^{*A}$.
	Hence we have completed the proof of sequential rationality. 
	
\end{proof}

\end{document}